\documentclass[12pt]{iopart}

\usepackage{iopams}

\expandafter\let\csname equation*\endcsname\relax

\expandafter\let\csname endequation*\endcsname\relax

\usepackage{amsmath}
\usepackage{graphicx}
\usepackage{amssymb}
\usepackage{amsopn}
\usepackage{amsthm}
\usepackage{epstopdf}
\usepackage{microtype}
\usepackage{graphicx}
\usepackage{amssymb}
\usepackage{amsfonts}
\usepackage{epstopdf}
\usepackage[hidelinks,hypertexnames=false]{hyperref}
\usepackage{comment}
\usepackage{enumerate}
\usepackage{float}
\usepackage{tikz}
\usepackage{pgfplots}
\pgfplotsset{compat=1.17}
\usepackage{framed}
\usepackage{xcolor}

\usepackage{soul}

\providecommand{\reals}{\mathbb{R}}

\providecommand{\eps}{\varepsilon}
\providecommand{\diff}{\mathrm{d}}

\providecommand{\x}{\mathbf{x}}

\providecommand{\argmax}{\operatornamewithlimits{\arg\max}}
\providecommand{\argmin}{\operatornamewithlimits{\arg\min}}
\renewcommand{\iota}{\textsf{i}}

\newtheorem{dfn}{Definition}
\newtheorem{prop}{Proposition}
\newtheorem{thm}{Theorem}

\newtheorem{cor}{Corollary}
\theoremstyle{remark}
\newtheorem{rmk}{Remark}

\usepackage[backend=bibtex,style=authoryear,natbib=true, url=false,isbn=false, doi=false, eprint=false]{biblatex} 
\begin{document}

\title[]{On the discretization of the object space in inverse problems with application to cryo-electron microscopy}

\author{G. Mordant$^\dagger$, L. Evans$^{*}$, D. Silva-S{\'a}nchez$^{\dagger}$, P. Cossio$^{*,\circ}$, R. Lederman$^{\dagger}$ }

\address{
$^{*}$ Center for Computational Mathematics, Flatiron Institute, New York, United States \\
$^{\circ}$ Center for Computational Biology, Flatiron Institute, New York, United States \\
$^{\dagger}$ Yale University, CT, USA}
\vspace{10pt}
\begin{indented}
\item[]\today
\end{indented}

\begin{abstract}

In many inverse problems, the aim is to recover a probability distribution
on a latent (or object state) space from indirect, noisy observations. When the observations
can be modelled as noisy samples from the mixing latent distribution, the recovery
problem is a deconvolution problem on the space of probability measures.
A common strategy is to fix a finite set of candidate points and estimate
a weight for each, turning the problem into a finite-dimensional concave
maximum likelihood problem on the simplex. We study the combined effect of
this discretization and of the noise in the context of cryo-electron
microscopy (cryo-EM), where the candidate points are biomolecular conformations
and the weights describe the
relative frequency of each conformation.
Our results pertain to both statistical and algorithmic aspects of the estimator.
We analyze the weight-recovery problem in which the candidate states and
their likelihoods are known. A nearby pair of candidates forces a
near-null direction. More generally, the grid and the noise level impose
a uniform lower bound on the achievable Kullback--Leibler divergence
between observation densities, even with infinite data. The finite-grid
estimator is asymptotically normal when its population target is in the
interior of the simplex; at a boundary target, its limit is a cone-projected
Gaussian. Finally, the exact proximal form of Expectation--Maximization leads
to a global high-noise comparison between an early iterate and a
KL-penalized likelihood, without a basin assumption or a linearization of the
recursion. Tests on synthetic images of the Hsp90 molecule illustrate
the theoretical findings and translate them into practical guidelines for
interpreting reweighted ensembles.

\end{abstract}

\section{Introduction}

An objective of science is describing accurately how and why certain
quantities or objects fluctuate. The prototypical example of interest for
this paper is understanding different conformations of a molecule, i.e., the
different states in which that molecule can be found in the nature and their probabilities. There
usually is a continuum of such states and, if one had direct observations,
the problem would be a distribution-estimation problem. We call this
distribution the \emph{ensemble distribution};  the standard usage in the
mixture-model and empirical-Bayes literature would be mixture distribution, and the space on which it
lives the \emph{latent space}.  For our purposes, the space is rather called the conformational space as latent space bears a different meaning in the cryo-EM community.

In the setting of interest, one does not have direct observations from the conformational space as  data is usually stained by noise or only indirectly
observed due to the imaging projection. Simplifying a bit, each observation is a noisy sample
from the mixing distribution, with Gaussian noise of known variance. The
density of the observations is therefore the convolution of the mixing
distribution with a Gaussian kernel, and recovering the mixing distribution
from the observations is a deconvolution problem on the space of probability
measures. As no parametric form for the mixing distribution usually exists,
various approaches (see Section~\ref{sec: relatedCryo} below for a review) approximate it by a discrete distribution, i.e., a weighted
sum of Dirac masses supported on a finite set of candidate points, as done by \citet{tang2023ensemble} in cryo-EM to cite but one.   The
weights are then estimated from the data, usually by maximum likelihood.
The candidate points can come from domain knowledge, from numerical
simulations such as molecular dynamics (MD), or from a regular grid. In Section~\ref{sec: relatedCryo}, we further discuss the importance of discretization in a variety of schemes.

This paper is about the combined effect of the Gaussian noise and of the
discretization step. Once the support is fixed, the maximum likelihood
problem becomes a finite-dimensional concave optimization on the simplex of
weights, which is computationally convenient.

We argue that the noise level and the geometry of the candidate points
interact in ways a practitioner should know about. How many candidate
points there are, how far apart they sit, and how well they cover the
conformational space, together with the noise level $\sigma$, control three things
at once: the conditioning of the optimization problem, the best
KL approximation of the observation density that can be achieved
even with infinite data, and the
way reconstructed weights relate to the empirical frequencies (``populations'' in the cryo-EM literature) one would
estimate from the data. In a cryo-EM context this means, concretely, that two
candidate structures which are close compared to the noise level can
exchange weight almost freely without changing the likelihood, that an
MD run which fails to explore a region of the landscape
produces a population observation-density KL floor that more data cannot
lower, and that
large reconstructed weights should not be read directly as large
conformational probabilities. Quantifying these statements is the goal of
the paper.

More precisely, we explain
\begin{itemize}
    \item why a nearby pair of candidates forces a nearly flat
    direction for the optimization, at a scale controlled by their spacing relative to the noise;
    \item why the candidate points (obtained for instance from MD), together with the noise level, influence the best achievable
     KL divergence, even with infinite data;
    \item why sparse fitted weights need not correspond to large conformational probabilities;
    \item what the statistical properties of the estimator are.
\end{itemize}

The main theoretical results include the uniform fixed-grid bound of
Theorem~\ref{thm: LowBd}, the global early-EM comparison of
Proposition~\ref{prop: EqReg}, the additional-noise stability result of
Theorem~\ref{thm: NoiseMisspecification}, and the constrained limit
distribution of Theorem~\ref{thm: CLT}.

Besides these statistical aspects, the solver used can subtly influence the outcome.   The Expectation-Maximization (EM) algorithm is a popular way to solve this
optimization problem, even though off-the-shelf convex solvers could also provide a solution. The EM algorithm is easy to implement and has a small computational cost per iteration.

The maximum-likelihood objective is concave. The fitted vector $K\widehat\alpha$ is unique,
and the weights are unique under the tangent-injectivity condition introduced
below. The estimate can nevertheless be sparse or spiky; see 
\citet[Section~3]{kim2020fast} or panel f.\ of Figure~\ref{fig: EM}.

We explain why EM can become slow when components are similar or the
optimum lies on the boundary. From a positive initialization it stays in the
relative interior, and in the high-noise regime its proximal interpretation
makes precise how early stopping can regularize the weights. We also study
practical stopping diagnostics.

\textbf{Organization.}
We start with a more formal description of the problem that we tackle in Section~\ref{sec: Problem}.
Then, Section~\ref{sec: Background} provides additional background: we introduce cryo-EM
(Sec.~\ref{sec: cryo-EM}), the statistical-physics view of conformational
distributions (Sec.~\ref{sec: SimConf}), the ensemble-reweighting scheme of
\citet{tang2023ensemble} and its connection with nonparametric maximum
likelihood (Sec.~\ref{sec: EnsRew}), and related work in cryo-EM
(Sec.~\ref{sec: relatedCryo}). Section~\ref{sec: Obj} analyses the objective
function: we compute its gradient and Hessian and quantify the ill-conditioning
caused by close structures (Sec.~\ref{sec: GradHess}),
establish self-concordance, establish a fundamental KL lower bound when the grid fails to cover the conformational space (Sec.~\ref{sec: LimGrids}),  and close with an exact clusterwise description of how empirical frequencies enter the objective (Sec.~\ref{sec: Replica}). Section~\ref{sec: EMSect} is devoted
to the Expectation--Maximization algorithm used in practice: we explain
when early stopping can induce an implicit regularization, give a proximal-point
interpretation, and discuss stopping criteria. Section~\ref{sec: stats}
studies stability under an additional noise component and a central limit theorem
for the estimated weights~$\hat\alpha$. Section~\ref{sec: sim} illustrates the theoretical findings
through simulations on the Hsp90 molecule. We finalize with a conclusions Section \ref{sec: conclu}. The remaining proofs are deferred
to~\ref{sec: Proofs}. An additional example illustrating theoretical findings, the IgG molecule, is provided in \ref{sec: igg}.

\section{Problem formulation}
\label{sec: Problem}

In this section, we formally introduce the problem and fix notation for the rest of the paper. A mixing distribution with
density $\rho:\reals^d\to\reals$ is given, and $n \in \mathbb{N}$ i.i.d.\ sample points
$\mathfrak{x}_1,\ldots,\mathfrak{x}_n$ are drawn from it. One does not
observe the $\mathfrak{x}_i$ directly, but rather
$y_i=\mathfrak{x}_i+\eps_i$, where the $\eps_i$ are independent Gaussian
random vectors with covariance $\sigma^2 I_d$. This is the weight-recovery
subproblem obtained when the forward operators and candidate structures are
treated as known; we expect the same mechanisms to remain relevant in the full
cryo-EM problem.

Write $p_\sigma$ for the density of $\mathcal{N}(x,\sigma^2 I_d)$.
In the simplified model we consider, each $y_i$ has density
\[
  \int_{\mathbb{R}^d} p_\sigma(y \mid x)\,\rho(x)\,dx,
\]
the convolution of $\rho$ with the Gaussian kernel.

The measure $\rho$ will be estimated by a discrete measure of the form
$\sum_{m=1}^M \alpha_m\delta_{x_m}$, where the points $x_1,\ldots,x_M$ are
given and $\alpha$ lies in the simplex
$\Delta_{M-1}:=\{\alpha\in\reals^M:\alpha_m\geq 0,\ \sum_m\alpha_m=1\}$.
The associated maximum likelihood problem is
\begin{equation}
\label{eq: MainOpt}
\max_{\alpha\in\Delta_{M-1}} F(\alpha),
\qquad
F(\alpha):=\frac1n\sum_{i=1}^n\log\!\left(\sum_{m=1}^M p_\sigma(y_i\vert x_m)\,\alpha_m\right).
\end{equation}
Fixing the points $x_m$ rather than optimizing over them is common practice in integrative structural biology methods that rely on reweighting structural ensembles against experimental data (\cite{bottaro2020integrating}). Finding the optimal $\alpha$
then amounts to reweighting the structures. 
In the running example of this work, each $x_m$ represents a conformation of Hsp90. The latter are indexed by an opening angle $\theta$ and the $\alpha_m$'s would be the relative proportions.
General background on this
application is provided in Section~\ref{sec: Background} and the particular running example is shown on Figure~\ref{fig: EM} and developed in  Section \ref{sec: RunExp}.

\section{Background}
\label{sec: Background}
The goal of this paper is to analyze the properties and limitations of a previously introduced denoising method \citep{tang2023ensemble, evans2026counting} for cryo-EM data.
This method leverages noisy experimental data as well as output of numerical simulations, such as MD, and combines both sources with the objective of retrieving the most accurate representation of the underlying probability density of the molecules.
Though our focus is on the reweighting scheme, our results extend beyond the method analyzed here, and we give related work in Section~\ref{sec: relatedCryo}.

\subsection{Cryo-electron microscopy}
\label{sec: cryo-EM}
Biomolecules undergo changes between various conformations to perform different functions.
An important research problem is to understand the various conformations a biomolecule can take and subsequently infer the underlying biophysical mechanisms. These conformations can be studied through cryo-EM, a 
leading technique for imaging biomolecules at near atomic resolution \citep{nakane2020single}. 
This method has become a fundamental scientific tool: the 2017 Nobel prize in chemistry was awarded to Dubochet, Frank and Henderson for the development of cryo-EM, and the experiment was labeled 
a ``Method to Watch"  by Nature Methods in 2022 \citep{doerr2022dynamic}. 
In cryo-EM, a solution of many copies of biomolecule is flash frozen and imaged with a transmission electron microscope.
Due to fast freezing, the biomolecules remain in a near-native state and the resulting images capture noisy snapshots of various conformations. This contrasts with crystallization-based techniques that restrict the conformational variability.

Analyzing the acquired images is very challenging. 
As the molecules are flash-frozen, each is in a random unknown orientation. Also, one only observes a two-dimensional view of the 3D volume.  Finally, the noise level is extremely high, which makes the problem of recovering the conformation landscape, i.e., the set of possible conformations, more intricate. Usually, the image formation process is modeled as in equation~\eqref{eq: ForwMod}.

Our analysis is motivated by the forward model for heterogeneous cryo-EM datasets. In this forward model, one considers a dataset of $n$ images, where each image $y_i$ is a projection of an (unknown) volume $V_i$ corresponding to the electrostatic potential of the molecule of interest in the $i$-th conformation.   
The  collection of acquired images $\{y_i\}_{i=1}^n$ is then modeled as
\begin{equation}
\label{eq: ForwMod}
y_i =(PSF_i \circ T_i \circ P \circ R_i)V_i +\eta_i, \qquad i \in \{1,\ldots, n\},
\end{equation}
 where $R_i$ is a 3D rotation operator corresponding to the random orientation of the volume $V_i$, $P$ is the 2D projection operator and $T_i$ is a 2D translation operator corresponding to the offset of the projected volume with respect to the center of the image. Both $R_i$ and $T_i$ are specific to each particle image $y_i$ and are unknown. $PSF_i$ is the Point Spread Function operator applied to the projected image and can be thought of as a convolution coming from the fact that the microscope is imperfect and distorts the signal.  Finally, $\eta_i$ is  some random noise, often assumed to be Gaussian, the covariance matrix of which we do not specify for now. 
 For initial processing of an image dataset, $PSF_i, T_i, R_i$ are all unknown in addition to $V_i$ for $i=1,\ldots, n$, and must be estimated, often through alternating minimization using a likelihood similar to~\eqref{eq: imaging_likelihood} \citep{singer2020computational}.

The composition of the various functions appearing in \eqref{eq: ForwMod} is called the forward operator, which we denote   $\operatorname{Img}_{T_i, R_i}(\cdot)$.
In practice, the conditional density $p_\sigma(y_i\vert\operatorname{Img}_{T_i, R_i}(V))$  for a volume $V$, usually is taken to be Gaussian, i.e.,
\begin{equation}
\label{eq: imaging_likelihood}
	p_\sigma\big(y_i\vert \operatorname{Img}_{T_i, R_i}(V)\big) := \big(2\pi\sigma^2\big)^{-N/2} \exp\left(- \frac{\|y_i -\operatorname{Img}_{T_i, R_i}(V)\|^2}{2\sigma^2}\right),
\end{equation}
where $N$ is the number of pixels, and $\sigma$ is the fixed noise level, which we assume is known.
We note that we take the cryo-EM forward operator~\eqref{eq: ForwMod} and likelihood~\eqref{eq: imaging_likelihood}  as our inspiration, but for the analysis we will consider an identity forward operator and Gaussian likelihood with large variance, as explained in Section \ref{sec: EnsRew}.

\subsection{Modeling the conformational ensemble distribution}
\label{sec: SimConf}

 A fundamental idea of statistical physics is that the variations of the configuration of a molecule, i.e., the different conformations should depend on the energy level of the configuration and the temperature. That is, one expects to observe only configurations that are stable and thus have a low energy level. At higher temperatures, as the system gets more excited,  more variability can be expected.
The probability distribution underlying the configurations is modeled as the Boltzmann distribution 
\[
\rho(x) = \frac{1}{Z_0} e^{-\beta \mathcal{H}(x)},
\]
where $\mathcal{H}$ is the molecular Hamiltonian characterizing the energy, $Z_0$ is the normalizing constant and $\beta$ the inverse temperature parameter. Unfortunately, the Hamiltonian is not known exactly, which is why understanding of molecular conformations is still an active topic of research.

In parallel to the experiments and the data they deliver, computational advances from MD \citep{jung2023machine} or AlphaFold \citep{abramson2024accurate} for instance, enable one to predict such conformations. A natural  direction is then to combine the strength of those numerical methods with the datasets available to obtain the most accurate reconstructions possible and eventually improve the understanding of the conformational landscape. This strand of ideas is known as integrative biology, see  \citet{ward2013integrative,bottaro2018biophysical}. This combination is  important, as the computational methods are only providing approximations of the unknown truth and might exhibit biases.

\subsection{Ensemble reweighting as nonparametric maximum likelihood }
\label{sec: EnsRew}

An integrative approach as the one above is put forward in \citet{tang2023ensemble} and subsequent works \citep{silva2026cryo}. Their procedure is the one that we focus on in this paper, even though some of the phenomena we highlight are relevant for other types of analyses. Let us stress that we get rid of the issue of the rotations and projections for the analyses of the present work.

\citet{tang2023ensemble} propose the approach below. For the analysis
conducted in this paper, we work with a simplified version of their setting:
the cryo-EM forward operator of equation~\eqref{eq: ForwMod} is replaced by
the identity, the random rotations, translations and  PSF are discarded, and
the noise is assumed to be isotropic Gaussian with known variance~$\sigma^2$.
This simplification lets us isolate the statistical properties of the
reweighting step itself, independently of the pose-estimation and PSF
calibration problems, and study ill-conditioning, early-stopping
regularization, and spike formation within this subproblem.
We assume that an i.i.d.\ sample of images can be produced, call it
$Y:=\{y_i\}_{i=1}^n$. For simplicity, assume again that the images are
vectorized and live on $\mathbb{R}^N$ where $N$ is the total number of pixels.

Let us introduce the initial candidate probability density of the conformations $\rho_0(x)$, which could come from a model or any educated guess. Then, to incorporate the information contained in the data, a class of functions $h(x; \alpha)$ depending on $\alpha$ is introduced. The $\alpha$ parameter will be learnt so that $h(x; \alpha)\rho_0(x)$ describes more accurately the conformational space. One can then write the conditional likelihood of the data given the reweighting function as
\[
p_\sigma(Y\vert \alpha) = \prod_{i=1}^n \left(\int  p_\sigma(y_i\vert x) h (x; \alpha) \rho_0(x)\diff x \right).
\]
As a parametric form for  $h(x; \alpha)$ is not to be expected,  the authors propose to consider a discrete counterpart of the likelihood, which reads
\begin{equation}
\label{eq: PbInterest}
\tilde p_\sigma(Y\vert \alpha) = \prod_{i=1}^n \left( \sum_{m =1}^M p_\sigma(y_i\vert x_m) \alpha_m \right), 
\end{equation}
where the vector $\alpha$ lies in the $(M-1)$-simplex $\Delta_{M-1}$ and $\{x_m\}_{m=1}^M$ is a representative subset of the possible conformations, for instance a discrete conformational path, which will be obtained from MD simulations. The weight vector $\alpha$ now plays the role of the reweighting density $h(x;\alpha)$ above.
Finally note that 
\begin{equation}
\label{eq: Objective}
	\hat\alpha_n\in \argmax_{\alpha \in \Delta_{M-1} } \tilde p_\sigma(Y\vert \alpha)  = \argmax_{\alpha \in \Delta_{M-1} }\frac1n  \sum_{i=1}^n  \log \left( \sum_{m =1}^M p_\sigma(y_i\vert x_m) \alpha_m \right),
\end{equation}
which is convenient in many cases as noted by the authors.

Let us take a step back and start again from the fact that no parametric form for $h(x,\alpha)$ could be postulated. Indeed, no explicit form of an $\alpha$-parametrized Hamiltonian $\mathcal{H}_\alpha(x)$ is usually available or sufficiently accurate\footnote{Observe that there is no need to learn anything in the case of a conformational landscape for which the Hamiltonian is perfectly known.}, so that a parametrization of the form $h(x; \alpha)\rho_0(x) \propto e^{-\mathcal{H}_\alpha(x)}$ is not feasible. Still, following the principle of maximum likelihood, one could then maximize the log-likelihood over all possible distributions, i.e., 
\begin{equation}
\label{eq: NPMLE}
\max_{q \in \mathcal{P}}\ \frac 1n \sum_{i=1}^n \log \int p_\sigma( y_i \vert u) \diff q(u),
\end{equation}
where  $\mathcal{P}$ is the set of probability distributions supported on $\reals^{N}$. This problem
 is classical in statistics, known as the \emph{nonparametric maximum likelihood estimation problem}, which at least dates back to \citet{robbins1950generalization}. A book-long exposition of the subject is \citet{lindsay1995mixture} and the interested reader might profit from Chapters 1, 5 and 6 in particular.

This problem has been studied much over the past decades and keeps on being an active research field. We refer to the lecture notes by~\citet{ignatiadis2024empirical}.
An important fact is that the objective functional is concave in $q$. As it is an infinite dimensional problem, however, modifications have to be performed to make the problem computationally-friendly. 
Following \citet{laird1978nonparametric}, \citet{jiang2009general} considered the idea of restricting the set of measures in the optimization problem above by fixing the support of $q$ to a (large) discrete set of points, say $\{x_m\}_{m=1}^M$. With the set of points fixed, maximizing over all measures is then equivalent to optimizing the measures weights as is suggested above. 

This is the first takeaway: fixing the support to a discrete grid is a classical and principled way of making the NPMLE problem tractable. In our setting, this choice is not made freely: the grid $\{x_m\}$ is determined by MD simulations, which is precisely what makes the reweighting problem an instance of integrative biology instead of a purely statistical question.

Having to optimize over the weights only is a great advantage: the log-likelihood problem~\eqref{eq: Objective} is concave and finite-dimensional; classical optimization tools can thus be used.
To that end, \citet{koenker2014convex} proposed to solve that problem using duality and an interior point solver, for instance. The more recent algorithm by \citet{kim2020fast} reports better performance in the case $n \gg M$.
Finally, \citet{zhang2024efficient} further improved on the computations.

\subsection{Discretization of conformational space and related work in cryo-EM}
\label{sec: relatedCryo}

Within cryo-EM workflows, the reweighting and discretization updates analyzed here serve as foundational building blocks for several existing algorithms and pipelines. Both in cryoSPARC~\citep{punjani2017cryosparc} and RELION~\citep{scheres2012relion},  two  dominant software tools in the field, the conformational space is discretized, although structures are not given.  In the Flatiron Challenge for Cryo-EM Conformational Heterogeneity \citep{astore2025inaugural}, the participants were required to submit a discrete set of structures, i.e., provide a discrete estimate of the conformational space.
In multi-class 
expectation maximization such as in FREALIGN~\citep{lyumkisLikelihoodbasedClassificationCryoEM2013}, the scheme alternates between the reweighting update and an optimization step for the volumes of 3D classes and corresponds to soft assignment of particles to structures. The relationship to hard assignment is discussed in \cite{evans2026counting}.  
Similarly, the reweighting scheme mentioned above corresponds to an E step in another EM algorithm used to learn both structures and denoise at the same time, see for instance \citet{balanov2025expectation}. Thus, our results inform that scheme as well, even though capturing only part of the subtleties recently put forward. 
We further note that other empirical Bayes problems arise in cryo-EM:~\cite{gilles2025cryo} apply a PCA approach which induces a hierarchical linear model with heteroscedastic noise on latent cryo-EM volumes.

 Related work also considers Bayesian approaches to compare MD structures with cryo-EM images using the same likelihood as~\eqref{eq: PbInterest} \citep{tang2023ensemble, giraldo2021bayesian} or focusing on calculating \eqref{eq: imaging_likelihood} more efficiently when the pose is unknown \cite{dingeldein2025amortized}. More recently,~\cite{mattingly2026measurement} introduce an information-theoretic framework to determine the optimal discretization of conformational space in cryo-EM ensemble reweighting. By maximizing mutual information between structural weights and noisy simulated images, it is shown that measurement noise sets a limit on the ensemble resolution, enabling the selection of representative structures that capture molecular heterogeneity without spatial redundancy.

From the inverse-problems viewpoint, the fixed candidate set is a
discretization of the unknown measure, and the EM iteration count can act as a
smoothing or iterative-regularization parameter; Proposition~\ref{prop: EqReg}
quantifies one high-noise sense in which this occurs. Discretization effects have long been
studied in statistical inverse problems \citep{johnstone1991discretization},
while the Richardson--Lucy literature gives a close precedent for early
stopping \citep{bertero2009image}. Our focus is the interaction of these two
effects in the finite-grid weight problem.

\section{Properties of the MLE objective function}
\label{sec: Obj}

As discussed above, the reweighting problem is an optimization problem. It is therefore fundamental to study the properties of the objective function. 

\subsection{Gradient, Hessian and optimality condition.}
\label{sec: GradHess}
The objective function $F(\alpha)$ has a specific structure that determines 
computational tractability. We now characterize its curvature properties, 
which govern both optimization difficulty and statistical efficiency. Set $K_{ij}:=p_\sigma(y_i\mid x_j)$,
$s(\alpha):=K\alpha$, and $\mathcal T:=1_M^\perp$.

\begin{prop}
\label{prop: GradHessKKT}
For every $\alpha\in\Delta_{M-1}$,
\[
\nabla F(\alpha)=\frac1nK^\top s(\alpha)^{-1},
\]
and
\begin{equation}
\label{eq: Hess}
\nabla^2F(\alpha)
=-\frac1nK^\top\operatorname{diag}\!\big(s_i(\alpha)^{-2}\big)K,
\end{equation}
where the inverse is understood componentwise. Thus $F$ is concave, and it is strictly
concave on the simplex when
\begin{equation}
\label{eq:tangent-injectivity}
\ker(K)\cap\mathcal T=\{0\}.
\end{equation}
In that case the maximizer is unique; otherwise, the fitted vector
$K\hat\alpha$ is still unique.
\end{prop}
Note that condition~\eqref{eq:tangent-injectivity} can hold only if $M-1\leq n$. Moreover, it holds that
\begin{equation}
\label{eq:gradient-simplex-identity}
\langle\nabla F(\alpha),\alpha\rangle
=1.
\end{equation}
Hence the KKT conditions are
\begin{equation}
\label{eq: KKT}
\nabla F(\hat\alpha)-1_M+\hat\mu=0, 
\qquad \hat\mu\geq0,
\qquad \hat\alpha\geq0,
\qquad \hat\mu_j\hat\alpha_j=0,
\qquad 1_M^\top\hat\alpha=1.
\end{equation}
Equivalently, $[\nabla F(\hat\alpha)]_j=1$ if $\hat\alpha_j>0$, while $[\nabla F(\hat\alpha)]_j\leq1$ if $\hat\alpha_j=0$.
The following proposition quantifies the loss of curvature when two structures are close. Define the smallest tangent curvature by
\[
\lambda_{\min}^{\mathcal T}\!\bigl(-\nabla^2F(\alpha)\bigr)
:=
\min_{\substack{v\in\mathcal T\\ \|v\|=1}}
v^\top[-\nabla^2F(\alpha)]v.
\]

\begin{prop}
\label{prop: IllCond}
For distinct candidate indices $k$ and $\ell$, let $F_h$, $h\in\mathbb R^d$,
denote the objective in~\eqref{eq: MainOpt} obtained by replacing $x_k$ with
$x_\ell+h$, while keeping the data, $\sigma$, and the other candidate points
fixed. For the Gaussian kernel,
\[
\sup_{\alpha\in\Delta_{M-1}}
\lambda_{\min}^{\mathcal T}\!\bigl(-\nabla^2F_h(\alpha)\bigr)
=O(\|h\|^2),
\qquad h\to0.
\]
\end{prop}
This implies that weight can be shifted between two close structures with negligible
effect on the objective.

Under the high-noise scaling $Y_i=\mu_i+\sigma Z_i$, with
$Z_i\stackrel{\mathrm{iid}}{\sim}\mathcal N(0,I_d)$ and $\mu_i$, $x_\ell$,
and the remaining candidate points fixed, the same left-hand side is
$O_p(\|h\|^2/\sigma^2)$ as $\sigma\to\infty$, uniformly for $h$ in a fixed
neighborhood of zero.

\subsection{Self-concordance of the objective}

Importantly, the objective function exhibits a special structure called self-concordance. Self-concordance is a particularly convenient property for the analysis of the Newton algorithm, see for example~\citet{boyd2004convex}.

\begin{dfn}
\label{dfn: SelfConcordance}
Let $\Omega\subseteq\mathbb R^M$ be open and convex. A convex,
three-times differentiable function $f:\Omega\to\mathbb R$ is called
$M_f$-self-concordant if
\[
|D^3f(x)[u,u,u]|
\leq
M_f\bigl(D^2f(x)[u,u]\bigr)^{3/2}
\]
for every $x\in\Omega$ and $u\in\mathbb R^M$. The case $M_f=2$ is the
standard normalization.
\end{dfn}

\begin{prop}
\label{prop: Self-Conc}
Let
\[
\Omega:=\big\{\alpha\in\mathbb R^M:
K_{i,\bullet}^\top\alpha>0,\ i=1,\ldots,n\big\}.
\]
The negative objective $f(\alpha):=-F(\alpha)$ is
$2\sqrt n$-self-concordant on $\Omega$. Equivalently, $-nF$ is standard
self-concordant. Its Hessian is positive definite on the simplex tangent
space exactly under~\eqref{eq:tangent-injectivity}.
\end{prop}

\subsection{Fundamental limits of grid-based estimation}
\label{sec: LimGrids}

We first quantify the population likelihood gap created when the candidate
points cannot match the target moments.

For $\alpha\in\Delta_{M-1}$, set
$\mathrm P_M^\alpha:=\sum_m\alpha_m\delta_{x_m}$.
At the population level,
\[
\int (p_\sigma*\rho)(y)\log (p_\sigma*\mathrm P_M^\alpha)(y)\,\diff y
=\int (p_\sigma*\rho)(y)\log (p_\sigma*\rho)(y)\,\diff y
-\operatorname{KL}\!\left(p_\sigma*\rho\,\middle\|\,
p_\sigma*\mathrm P_M^\alpha\right).
\]
Thus, the population maximum likelihood minimizes this KL divergence.

For
$\kappa=(\kappa_1,\ldots,\kappa_d)\in\mathbb N_0^d$, put
$|\kappa|=\sum_{i=1}^d\kappa_i$ and
$x^\kappa=\prod_{i=1}^d x_i^{\kappa_i}$.
 We further place ourselves in the high noise case, i.e., one should think of $\sigma^2$ large with respect to the signal in the images, which is what one expects in cryo-EM.

As we have seen, maximizing the objective function is equivalent to minimizing the Kullback--Leibler divergence. Even in the case of infinite data, 
the latter cannot be zero unless 
\[
\rho \in \left\{\sum_{m=1}^M \alpha_m \delta_{x_m} : \alpha \in \Delta_{M-1}\right\}.
\]

 Let us also define
\begin{equation}
\label{eq: NumMom}
\begin{aligned}
\mathfrak m:=\sup\bigg\{\mathcal K\in\mathbb N_0:\ &
\exists\alpha\in\Delta_{M-1}\ \text{such that}\\[-2pt]
&\sum_{m=1}^M\alpha_m x_m^\kappa
=\int x^\kappa\,\diff\rho(x)
\ \text{for every }|\kappa|\leq\mathcal K
\bigg\}\in\mathbb N_0\cup\{+\infty\}.
\end{aligned}
\end{equation}

More intuitively, $\mathfrak m$ is the largest degree up to which all moments of $\rho$ can be matched simultaneously by $\mathrm P_M^\alpha$ for a suitable choice of $\alpha$.
Under the bounded-support assumption below, $\mathfrak m=+\infty$ if
and only if $\rho=\mathrm P_M^\alpha$ for some $\alpha$; this follows from
compactness of the simplex and moment determinacy on compact sets. The
theorem concerns the remaining case.

\begin{thm}
\label{thm: LowBd}
Suppose that $\rho$ has bounded support and let $\mathfrak m<\infty$ be
defined by~\eqref{eq: NumMom}. Then, for constants $c>0$ and
$\sigma_0<\infty$ depending only on $\rho$ and the candidate points,
\begin{equation}
\label{eq:uniform-kl-lower-bound}
\inf_{\alpha\in\Delta_{M-1}}
\operatorname{KL}(p_\sigma*\rho\|p_\sigma*\mathrm P_M^\alpha)
\geq
c\,\sigma^{-2(\mathfrak m+1)},
\qquad \sigma\geq\sigma_0.
\end{equation}
\end{thm}

The theorem shows that an incomplete grid leaves a nonzero approximation
error, even with infinite data, in the forward KL divergence between the true
observation density and the best convolved candidate mixture.

 Equation~\eqref{eq:uniform-kl-lower-bound} also reveals that a higher noise level $\sigma$ reduces the lower bound. This does not mean that the estimate gets better with higher noise level. As $\sigma\to\infty$, the signal is progressively lost and different weight vectors become increasingly difficult to distinguish.
One could say that ``if you squint your eyes enough, everything looks the same". Seen differently, the objective function gets flat.

 The moment-matching picture in Theorem~\ref{thm: LowBd} is closely related to
the low-SNR expansions of \citet{katsevich2023likelihood} and
\citet{fan2024maximum}, who study Gaussian mixture log-likelihoods with
unknown centers.
Our setting differs: the candidate points are fixed,
and the inference is over the simplex of weights. The same
moment-matching obstruction reappears, now in terms of which moments of
$\rho$ the fixed grid can reproduce.
For a fixed pair of Gaussian-smoothed measures,
\citet{chen2021asymptotics} obtain exact high-noise KL asymptotics. The point
of Theorem~\ref{thm: LowBd} is different: its one-sided bound is uniform over
the simplex, so even weights allowed to depend on $\sigma$ cannot remove the
moment obstruction imposed by the fixed grid.

The practical implication is the following. 
As there are
$\binom{d+\mathfrak m}{d}-1$ nonconstant coordinate moments through order
$\mathfrak m$ and since an $M$-point measure has $M-1$ free weights, a generic
moment vector requires, for fixed $d$, $M$ of order $\mathfrak m^d$. This necessary
dimension-counting condition might not be a sufficient condition for a fixed grid.
It further holds in the case of infinite data. However, as the conditioning deteriorates with  structures close to one another, a smaller number of structures might deliver better results in finite sample applications. Initial work has confirmed this in the integrative structural biology setting, where \citet{mattingly2026measurement} sub-sample structures through clustering procedures that stabilize quantities related to the Hessian.

\subsection{Empirical frequencies and candidate geometry}
\label{sec: Replica}

The optimizer can lie on the boundary and be spiky. We next consider a
correctly specified on-grid model and describe how empirical frequencies and
candidate geometry enter the objective.

To understand the role of empirical frequencies, condition on labels
$z_i\in\{1,\ldots,M\}$ and suppose that the observations are independent with
$Y_i\mid z_i=k\sim\mathcal N(x_k,\sigma^2I_d)$. Set
$N_k:=\#\{i:z_i=k\}$ and  $R_x:=\max_{k,j}\|x_j-x_k\|$.

\begin{prop}[Clusterwise representation]
\label{prop: Replica}
Assume that $M,d$ and the candidate centers are fixed, and fix $\sigma>0$.
If, as $n\to\infty$, $N_k/n\to\pi_k>0$, then, conditionally on the labels,
uniformly in $\alpha\in\Delta_{M-1}$,
\begin{equation}
\label{eq:clusterwise-kl-representation}
F(\alpha)
=-\frac d2\{1+\log(2\pi\sigma^2)\}
-\sum_{k=1}^M\frac{N_k}{n}
\operatorname{KL}\!\left(p_\sigma(\cdot\mid x_k)\| \
p_\sigma*\mathrm P_M^\alpha\right)
+O_p(n^{-1/2}).
\end{equation}
\end{prop}

\begin{rmk}[High-noise consequence]
The first-order case of \citet[Theorem~2.1]{katsevich2023likelihood},
applied with truth
$\delta_{x_k}$ and candidate measure $\mathrm P_M^\alpha$, gives, uniformly
in $k$ and $\alpha$,
\begin{equation}
\label{eq:clusterwise-kl-high-noise}
\operatorname{KL} \big (p_\sigma(\cdot\mid x_k)\|
p_\sigma*\mathrm P_M^\alpha\big)
=\frac{\|\sum_j\alpha_jx_j-x_k\|^2}{2\sigma^2}
+O_{M,d}\big(R_x^4\sigma^{-4}\big),\qquad \sigma\geq R_x.
\end{equation}
\end{rmk}

Writing $\bar x_n:=\sum_k(N_k/n)x_k$, combining the two displays gives, for
fixed $\sigma\geq R_x$ and up to terms independent of $\alpha$,
\[
F(\alpha)
=-\frac{1}{2\sigma^2}
\left\|\sum_j\alpha_jx_j-\bar x_n\right\|^2
+O_{M,d}(R_x^4\sigma^{-4})+O_p(n^{-1/2}).
\]
Thus the $\sigma^{-2}$ term is maximized by the weights whose barycenter is
$\bar x_n$. Empirical frequencies determine this barycenter; the candidate
geometry determines which weights share it.

 \section{Expectation-Maximization and early-stopping regularization}
 \label{sec: EMSect}
There is a long tradition of solving mixture problems with the expectation-maximization algorithm. The history goes back at least 50 years, with application in fields such as astrophysics, medical imaging, finance, and information theory, see for example~\citet{lucy1974iterative, vardi1993, cover1984}. We refer to \citet{redner1984mixture,balakrishnan2017statistical} for a complete introduction, the underlying motivation and further background. 

The EM algorithm updates the weight vectors

\begin{equation}
\label{eq: Upd}
\alpha_j^{(t+1)}
=\frac1n\sum_{i=1}^n
\frac{K_{ij}\alpha_j^{(t)}}
{\sum_{m=1}^M K_{im}\alpha_m^{(t)}},
\qquad t=0,1,\ldots .
\end{equation}
Write $u:=1_M/M$ and initialize at $\alpha^{(0)}=u$. We use $T$ for the
number of iterations. Define
\begin{equation}
\label{eq: Resp}
r_m(y;\alpha):=\frac{\alpha_mp_\sigma(y\mid x_m)}
{\sum_{\ell=1}^M\alpha_\ell p_\sigma(y\mid x_\ell)},\qquad
\Phi_n(\alpha):=\frac1n\sum_{i=1}^nr(y_i;\alpha), 
\end{equation}
so that the recursion reads $\alpha^{(t+1)}=\Phi_n(\alpha^{(t)})$.
This nonlinear recursion remains particularly challenging to study. Also note that $r$ are usually called \emph{responsibilities}.

This algorithm, in spite of its great popularity in practice, can be slow in some settings. Indeed,  \citet[p.~64]{lindsay1995mixture} stated: 
``However, one key feature of the algorithm is that it commonly displays a
very slow linear rate of convergence, where the rate constant is related to the
amount of information in the missing portion of the data. If the components
are similar in their densities, then the convergence is extremely slow. The
convergence will also be slow when the maximum likelihood solution requires
some of the weight parameters to be zero, because the algorithm can never
reach such a boundary point.''

 As hinted at in the Introduction, we aim at raising awareness in the community that there can be a regularization phenomenon through early stopping, as the EM might not be run long enough in the case of very small or zero optimal weights.  This might however be a feature for the applications of interest.  Recall that other fast convex-optimization methods exist and were discussed in Section~\ref{sec: EnsRew}; they can produce the same spiky boundary solutions in our examples (see Figure~\ref{fig: EM} panel f.). The stopping time along the EM path matters, as the spikiness of the converged fit is both a blessing and a curse. The converged fit indeed is adaptive, but might also introduce non-physical spikes. We believe that a properly early-stopped EM algorithm might provide the best of both worlds.

As detailed in Section~\ref{sec: RunExp}, we use the Hsp90 system and a simplified cryo-EM simulation setup as an example to recover a distribution over the chain opening angle $\theta$ (shown in blue in Figure~\ref{fig: EM}). That function is discretized and its atoms are used as the candidates to reweight. The figure follows EM towards the convex-optimization fit produced by SCS. The latter is sparser than the displayed target, which is not desirable in this example.

\begin{figure}[htbp!]
\begin{center}
\includegraphics[width=0.99\textwidth]{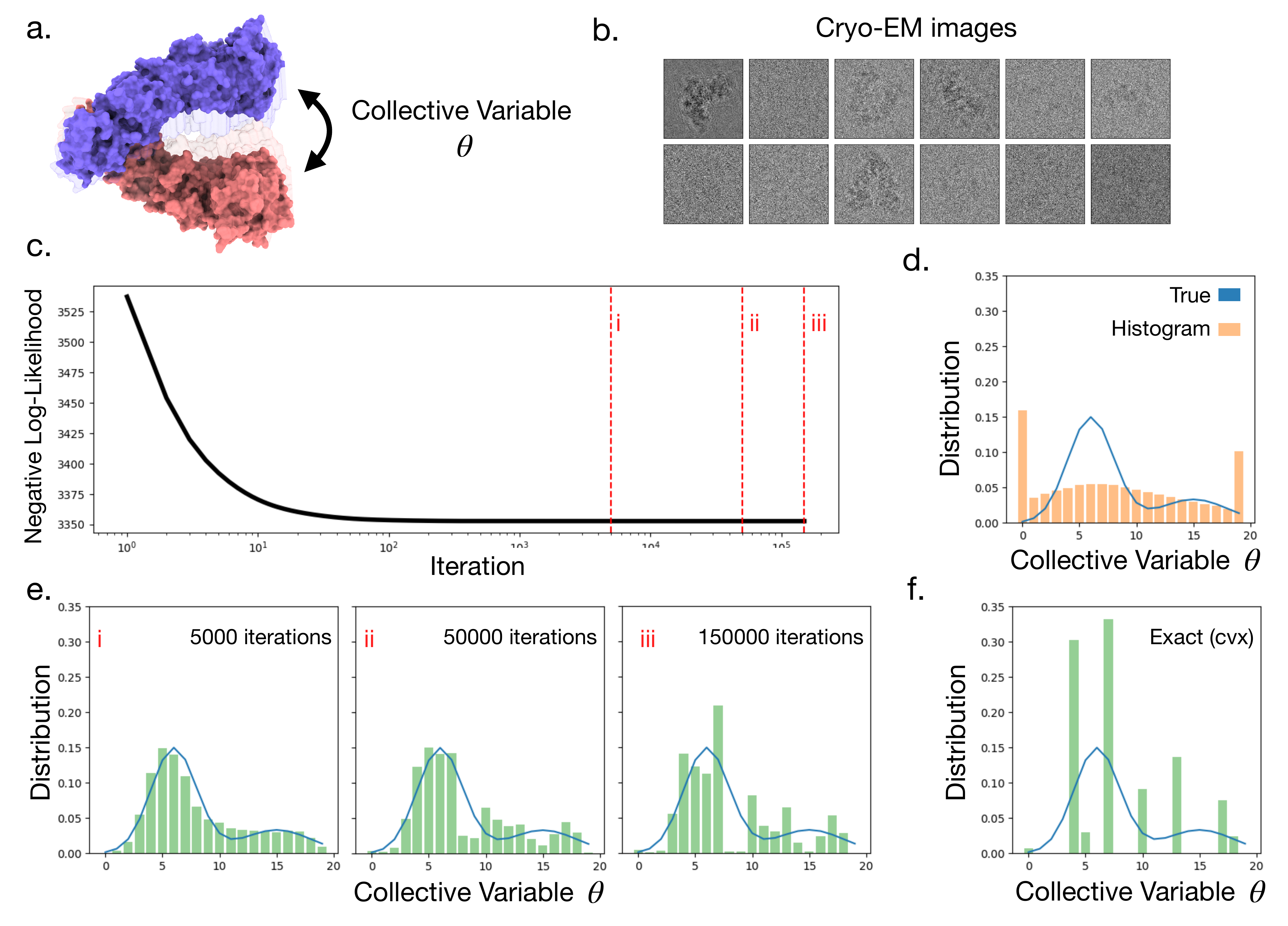}
\caption{\label{fig: EM}
The estimation of the weights as a function of the EM iteration,  illustrated on the cryo-EM Hsp90 example from Section~\ref{sec: sim}, with $n=100{,}000$ images and $M=20$ candidates, where the distribution is shown as a function of the chain opening collective variable $\theta$. 
The blue curve is the displayed truncated target density; normalized grid weights are used when a weight divergence is reported.
{\textbf a.} Molecule Hsp90 of interest for the cryo-EM simulations.  The motion simulated is the opening and closing of the angle between the two chains, illustrated by the transparent models and collective variable $\theta$.
{\textbf b.} Example of simulated cryo-EM images. 
{\bf c.} Evolution of a scaled negative log-likelihood as a function of the number of iterations (log scale).
{\bf d.} Estimation of the frequency (population) of each structure based on the number of nearest neighbors. 
{\bf e.} Estimated weights at iteration number 5000, 50,000 and 150,000, respectively, of the EM algorithm. The sparsity pattern evolves towards the SCS fit in panel~f.
{\bf f.} Convex-optimization fit produced by CVXPY/SCS with absolute and relative tolerances $10^{-5}$.
}
\end{center}
\end{figure}


Remark that the update can be rewritten

\begin{equation}
\label{eq: Updt}
\alpha^{(t+1)} =\alpha^{(t)} +\alpha^{(t)}\odot\left(\nabla F(\alpha^{(t)})-1_M\right),
\end{equation}
which is a diagonally scaled multiplicative gradient step.  At the first step, as $\alpha=1_M/M$, this is an ordinary gradient step of
size $1/M$ in the simplex tangent direction
$\nabla F(1_M/M)-1_M$.
Remark further that once an entry of the vector $\alpha$ is zero, it cannot become positive again.

\subsection{The proximal point algorithm view on early stopping regularization}

Another way of seeing the algorithm is as a proximal-point algorithm, as in
\citet{chretien2008algorithms}; see also \citet{keys2019proximal}. With the
responsibilities defined in \eqref{eq: Resp}, set
\[
Q(\alpha\mid\beta):=\frac1n\sum_{i,j}r_j(y_i;\beta)
\log(\alpha_jK_{ij}),
\qquad
\mathcal D(\alpha,\beta):=\frac1n\sum_i
\operatorname{KL}\big(r(y_i;\beta)\big\| r(y_i;\alpha)\big).
\]
We use the standard extended-value conventions $0\log0=0$,
$a\log0=-\infty$ for $a>0$, and the corresponding infinite-KL convention.

\medskip \medskip 

\begin{prop}
\label{prop: Prox}
The EM update is the unique maximizer of $Q(\alpha\mid\alpha^{(k)})$ and,
for a constant $C_k$ independent of $\alpha$,
\begin{equation}
\label{eq: Prox}
Q(\alpha\mid\alpha^{(k)})
=F(\alpha)-\mathcal D(\alpha,\alpha^{(k)})+C_k.
\end{equation}
Thus EM is an exact proximal-point iteration with a KL penalty between the
responsibilities at two successive weight vectors.
\end{prop}
Observe that the curvature of $Q$ in coordinate $j$ is
$-\alpha_j^{(k+1)}/\alpha_j^2$; it therefore becomes large near the boundary
when $\alpha_j^{(k+1)}>0$. This explains why the EM algorithm slows down. 

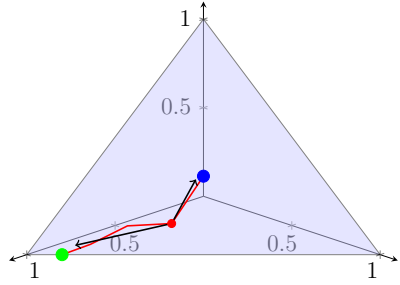
\begin{figure}[htbp!]
\begin{center}
\begin{tikzpicture}[scale=.75]
    \begin{axis}[
        view={135}{25},
        axis lines=middle,
        xlabel={},
        ylabel={},
        zlabel={},
        domain=0:1,
        samples=100,
        zmin=0, zmax=1.1,
        ymin=0, ymax=1.1,
        xmin=0, xmax=1.1
    ]
    \coordinate (A) at (1, 0, 0);
    \coordinate (B) at (0, 1, 0);
    \coordinate (C) at (0, 0, 1);
    
    \filldraw[fill=blue!20, draw=black, opacity=0.5] (A) -- (B) -- (C) -- cycle;
    
    \coordinate (center) at (1/3, 1/3, 1/3);
    \coordinate (centerb) at (0.362, 0.319, 0.320);
    \coordinate (end) at (0.9, 0.1, 0);
    \coordinate (endb) at (0.844 ,0.118 ,0.040);
   
    \coordinate (arrow_point) at  (0.475, 0.295, 0.1);
    
    \addplot3[red, thick] coordinates {
        (1/3, 1/3, 1/3)
        (0.475, 0.295, 0.1)
        (0.62, 0.19, 0.10)
        (0.8, 0.16, 0.045)
        (0.9, 0.1, 0)
    };
    
    \addplot3[only marks, mark=*, color=blue, mark size=3pt] coordinates {(1/3, 1/3, 1/3)};
    \addplot3[only marks, mark=*, color=green, mark size=3pt] coordinates {(0.9, 0.1, 0)};
    \addplot3[only marks, mark=*, color=red, mark size=2pt] coordinates {(0.475, 0.295, 0.1)};
    
    \draw[->, thick ] (arrow_point) -- (centerb);
    \draw[->, thick] (arrow_point) -- (endb);
    \end{axis}
\end{tikzpicture}
\end{center}
\caption{\label{fig: EM_visualized} Visualization of the trajectory of the EM algorithm (red) on the simplex (transparent blue), when initialized at the equal weight vector (blue point). The solution of the optimization problem is the point in  green. At each step, the algorithm balances the two terms in Equation~\eqref{eq: Prox}, which are represented by the black arrows. }
\end{figure}

Figure~\ref{fig: EM_visualized} illustrates this proximal update on a
three-dimensional simplex. On that figure, the red broken line represents the iteration path, the blue dot the starting vector of weights, the green dot the optimizer and the two black arrows represent the two competing terms in equation~\eqref{eq: Prox}.

There is no tunable step size: the nonnegative penalty changes with the
iterate and vanishes at the current one. This makes the implicit regularization of the scheme hard to parse in full generality.
The high-noise comparison in Proposition~\ref{prop: EqReg} below
makes the role of the iteration count as a regularization parameter precise.
To establish that result, let $Y_i^{(\sigma)}=\mu_i+\sigma Z_i$, with
$Z_i\stackrel{\mathrm{iid}}{\sim}\mathcal N(0,I_d)$. Recall $R_x$ and set
\[
A_x:=\max_{i,j}\|x_j-\mu_i\|,
\qquad
a_\sigma:=\frac{R_x}{\sigma}+\frac{A_xR_x}{\sigma^2}.
\]
For an integer $T\geq0$, let $\widetilde\alpha_{\sigma,T}$ be the unique
maximizer of
\[
F(\alpha)-\frac1{T+1}\operatorname{KL}(u\|\alpha)
\quad\text{over }\Delta_{M-1}.
\]

\begin{prop}[Early EM as KL regularization]
\label{prop: EqReg}
Fix $n,M,d$, the means $\mu_i$, and the candidate grid, with
$R_x>0$. If $T_\sigma$ is deterministic and
$(T_\sigma+1)a_\sigma\to0$, then
\begin{equation}
\label{eq:empirical-em-regularization-rate}
\frac{\|\alpha^{(T_\sigma+1)}-
\widetilde\alpha_{\sigma,T_\sigma}\|_2}
{\|\alpha^{(T_\sigma+1)}-u\|_2}
=O_p\bigl((T_\sigma+1)a_\sigma\bigr)=o_p(1).
\end{equation}
\end{prop}
Thus, $(T_\sigma+1)$ EM updates are, to first order, equivalent to maximizing the regularized objective
$F(\alpha)-(T_\sigma+1)^{-1}\operatorname{KL}(u\|\alpha)$.

\subsection{Early stopping and variability of the estimates}

Theorem~\ref{thm: CLT} describes the local fluctuations of the
finite-grid MLE, while  Proposition~\ref{prop: EqReg}  compares an
early-stopped EM iterate with a penalized optimizer.
The two effects are distinct: sampling changes the estimator, whereas
early stopping changes how well the objective is optimized. 
Intuitively, with less iterates, the iterates must remain closer to the uniform initialization and are therefore less variable but biased. 
On the other hand, a fully solved optimization problem comes with a high variability because of the flatness of the objective when candidate structures are close. 
Completely understanding and quantifying this phenomenon rigorously remains an interesting challenge.

\subsection{Stopping criteria} 

\begin{figure}[htbp!]
\begin{center}
\includegraphics[width=0.99\textwidth]{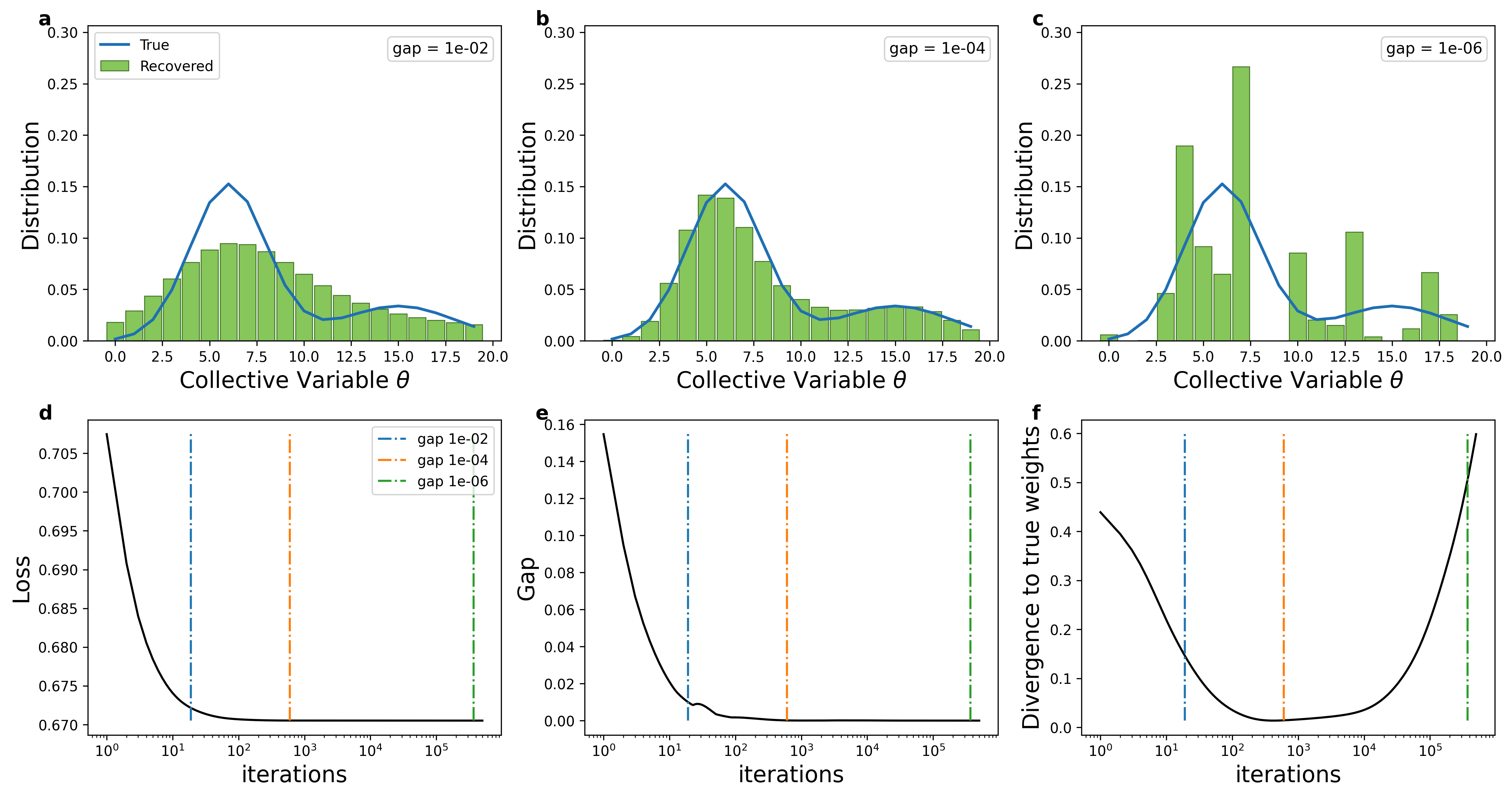}
\caption{\label{fig: EM_with_thresholds}
Comparison of thresholds for the interior fixed-point diagnostic
$\|\nabla F(\alpha^{(t)})-1_M\|_\infty$ on the Hsp90 example
($n=100{,}000$, $M=20$).
{\bf a.--c.} Estimated weights at thresholds $10^{-2},10^{-4},10^{-6}$.
{\bf d.--f.} Loss, diagnostic, and weight divergence
$\operatorname{KL}(\alpha^{(t)}\|\alpha_{\mathrm{true}})$ to the displayed
true weights. The saved stopping iterations are $19$, $603$, and $369458$.}
\end{center}
\end{figure}

In view of the KKT condition, the gap
\begin{equation}
\label{eq:FW-gap}
\mathcal G(\alpha):=
\max_{1\leq j\leq M}[\nabla F(\alpha)]_j-1.
\end{equation}
vanishes exactly at a maximizer, and concavity gives
\(
0\leq F(\hat\alpha_n)-F(\alpha)\leq\mathcal G(\alpha).
\)
Indeed, concavity bounds the difference by
$\max_{\beta\in\Delta_{M-1}}\langle\nabla F(\alpha),\beta-\alpha\rangle$,
which equals~\eqref{eq:FW-gap} by~\eqref{eq:gradient-simplex-identity}.
The infinity norm may still be useful empirically, but need not vanish at a
sparse solution.

For the EM step, write
\(
d_t:=\alpha^{(t+1)}-\alpha^{(t)}
\)
and define
\begin{equation}
\label{eq: stopCrit}
\mathcal V_t :=\left\langle\nabla F(\alpha^{(t)}),d_t\right\rangle.
\end{equation}
Equation~\eqref{eq: Updt} gives the useful identity
\[
\mathcal V_t=\sum_j\alpha_j^{(t)}
\bigl([\nabla F(\alpha^{(t)})]_j-1\bigr)^2.
\]
Thus $\mathcal V_t$ is the variance gap of \citet{chok2025optimization}, a
weighted squared norm that downweights coordinates approaching zero. For $\beta\in\Delta_{M-1}$, define the local seminorm
\(
\|v\|_\beta^2:=v^\top[-\nabla^2F(\beta)]v.
\)

\begin{prop}
\label{prop: SelfConcGrad}
Let $\alpha,\beta\in\Delta_{M-1}$, put $d=\alpha-\beta$, and assume
\[
\tau:=\sqrt n\,\|d\|_\beta<1.
\]
Then
\begin{equation}
\label{eq:self-concordant-gradient-bound}
(1+\tau)^{-1}\|d\|_\beta^2
\leq
\langle\nabla F(\beta)-\nabla F(\alpha),d\rangle
\leq
(1-\tau)^{-1}\|d\|_\beta^2.
\end{equation}
\end{prop}

\begin{cor}
\label{cor:cor}
Let $\tau_t:=\sqrt n\,\|d_t\|_{\alpha^{(t)}}<1$. Then
\begin{align*}
&\langle\nabla F(\alpha^{(t+1)}),d_t\rangle+(1+\tau_t)^{-1}\|d_t\|_{\alpha^{(t)}}^2
\leq\mathcal V_t\\
&\qquad\leq
\langle\nabla F(\alpha^{(t+1)}),d_t\rangle
+(1-\tau_t)^{-1} \|d_t\|_{\alpha^{(t)}}^2.
\end{align*}
\end{cor}

Thus, $\mathcal G$ is a certified optimization gap, while $\mathcal V_t$ is a
smoother EM diagnostic as we shall further show in Figure~\ref{fig: igg_gaps_comparison}.

\subsection{Practical considerations on stopping criteria}
For early stopping, we compare the variance-gap criterion with the
interior fixed-point diagnostic
$\|\nabla F(\alpha)-1_M\|_\infty<\eta$, using
$\eta\in[10^{-3},10^{-2}]$. Empirically, the square root of the variance gap
behaves like a smoothed version of this diagnostic, so we use
\[ 
\big\langle \nabla_\alpha F(\alpha^{(t)}), \alpha^{(t+1)}- \alpha^{(t)}\big\rangle^{1/2} < \eta
\]
with $\eta\in[10^{-3},10^{-2}]$.

The square root of the variance gap can be interpreted as a weighted $L_2$-norm of $\nabla F(\alpha^{(t)}) - 1$ with respect to the weights $\alpha^{(t)},$ and is also related to the time-derivative of $F$ for a continuous-time flow on the simplex~\citep{chok2025optimization}. Connecting these interpretations to our results, with the goal of stronger numerical guidance for practitioners, is an important focus for future development.

\section{Statistical properties}
\label{sec: stats}

\subsection{Stability under an additional noise
component}

Recall that the ambient image dimension is $N$. As above, $y_i$ is obtained from the signal plus noise model. 
Let $\eta_i$, $i=1,\ldots,n$, be noise vectors in $\mathbb{R}^N$ and set $\widetilde y_i=y_i+\eta_i$.  Let $\widetilde F$ denote the objective
computed from $\widetilde y_i$. Let $\widehat\alpha$ maximize $F$, and let
$\widetilde\alpha$ be any maximizer of $\widetilde F$. Set
\[
\mu_N:=
\inf_{\alpha\in\Delta_{M-1}}
\lambda_{\min}^{\mathcal T}\!\bigl(-\nabla^2F(\alpha)\bigr).
\]

\begin{thm}
\label{thm: NoiseMisspecification}
Assume that the collection $(\eta_i)_{i=1}^n$ is independent of the
uncorrupted data, that $\mathbb E\eta_i=0$, and that
$\mathbb E(\eta_i\eta_i^\top)=\Sigma_i$. Write
$\varepsilon_N:=\max_i\|\Sigma_i\|_{\mathrm{op}}^{1/2}$.
For fixed $n,M$, suppose~\eqref{eq:tangent-injectivity} holds almost surely.
Then, as $N\to\infty$,
\begin{equation}
\label{eq:misspecified-noise-argmax}
\|\widetilde\alpha-\widehat\alpha\|
=O_p\!\left(\frac{R_x\varepsilon_N}{\mu_N\sigma^2}\right).
\end{equation}
\end{thm}

Here, $X_N=O_p(A_N)$ for a positive random scale $A_N$ means
$X_N/A_N=O_p(1)$.
The random factor $\mu_N^{-1}$ records the conditioning of the uncorrupted
problem. In particular, their distance converges to zero whenever
$R_x\varepsilon_N/(\mu_N\sigma^2)\to0$ in probability. 

Therefore, for heteroscedastic noise that is not ``too spiky" in the sense that $\varepsilon_N$ is not too large, the estimates will be robust to noise.

\medskip
\subsection{Asymptotic distribution of the finite-grid}
We use the standard constrained $M$-estimation limit
\citep{geyer1994constrained}. We now return to the fixed ambient dimension
$d$ of Section~\ref{sec: Problem}, keep the finite grid fixed as
$n\to\infty$, and parametrize the simplex by
\[
\Theta:=\{\theta\in\mathbb R^{M-1}:\theta_j\geq0,\
\mathbf1^\top\theta\leq1\},
\qquad
\alpha(\theta):=(\theta_1,\ldots,\theta_{M-1},
1-\mathbf1^\top\theta).
\]
This parametrization is useful for the central limit theorem below as any element on the simplex has $M-1$ degrees of freedom.
Let $f_\theta:=p_\sigma*\mathrm P_M^{\alpha(\theta)}$ and
\[
s_\theta(y):=\nabla_\theta\log f_\theta(y)
=\frac{\bigl(p_\sigma(y\mid x_j)-p_\sigma(y\mid x_M)\bigr)_{j=1}^{M-1}}
{f_\theta(y)}.
\]
Let $\alpha^*$ be a maximizer of
$\mathbb E[\log\{(p_\sigma*\mathrm P_M^\alpha)(Y)\}]$ over
$\Delta_{M-1}$.
Let $\theta^*$ and $\hat\theta_n$ be the first $M-1$ coordinates of
$\alpha^*$ and $\hat\alpha_n$, respectively. Put
\[
I_{\theta^*}:=
\mathbb E[s_{\theta^*}(Y)s_{\theta^*}(Y)^\top],
\qquad
b:=\mathbb E[s_{\theta^*}(Y)].
\]
As the MLE is constrained, the constraint will be relevant for the limiting distribution. We thus introduce the critical cone 
\[
\mathcal C:=\left\{u\in\mathbb R^{M-1}:
\begin{array}{l}
u_j\geq0\text{ whenever }\theta_j^*=0,\\
\mathbf1^\top u\leq0\text{ whenever }\mathbf1^\top\theta^*=1,\\
b^\top u=0
\end{array}
\right\}.
\]

\begin{thm}
\label{thm: CLT}
Suppose that $P_0$ has density $p_\sigma*\rho$, where $\rho$ is a
boundedly supported probability measure, and that the candidate points are
pairwise distinct. Then $I_{\theta^*}$ is positive definite and
\begin{equation}
\label{eq:constrained-mle-limit}
\sqrt n(\hat\theta_n-\theta^*)
\Rightarrow
\argmin_{u\in\mathcal C}
\left\{\frac12u^\top I_{\theta^*}u-Z^\top u\right\},
\qquad
Z\sim\mathcal N_{M-1}(0,I_{\theta^*}-bb^\top),
\end{equation}
where $\Rightarrow$ denotes convergence in distribution.
If $\theta^*$ is in the interior of $\Theta$, the limit is
$\mathcal N_{M-1}(0,I_{\theta^*}^{-1})$.
\end{thm}

\begin{rmk}[Correct specification and boundary points]
If $P_0$ has density $f_{\theta^*}$, then $b=0$ and $\mathcal C$ is the
tangent cone. After multiplication by $I_{\theta^*}^{1/2}$, the limit in
\eqref{eq:constrained-mle-limit} is the Euclidean projection of a standard
Gaussian vector onto $I_{\theta^*}^{1/2}\mathcal C$. At a boundary point
this projection is generally non-Gaussian.
\end{rmk}

\begin{rmk}[Misspecification at the boundary]
At a boundary maximizer under misspecification, $b$ need not vanish. For a
feasible displacement $u/\sqrt n$, the full log-likelihood has leading
deterministic term $\sqrt n\,b^\top u$, while its random fluctuation is
$O_p(1)$. If $b^\top u<0$, the deterministic loss rules out that direction
asymptotically. Thus $\mathcal C$ retains only feasible directions satisfying
$b^\top u=0$.
\end{rmk}

\section{Simulation results}
\label{sec: sim}

In this section, we present numerical experiments illustrating the issues that were theoretically investigated above. In Section \ref{sec: RunExp}, we present the biomolecular system and cryo-EM setup that we used as an example. 
In Section~\ref{sec: ConfSpace}, we exhibit how the weights can be biased when the structures do not cover the space well enough, thereby illustrating the practical consequence of Theorem~\ref{thm: LowBd}.
In Section~\ref{sec: Transient}, the question is similar, but instead of having part of the conformational space not covered, the density of the structures is too low in a certain region. In Section~\ref{sec: NumExpSpikes}, we then numerically investigate what determines where the spike appears, echoing Proposition~\ref{prop: Replica}.

\subsection{A running example: Hsp90}
\label{sec: RunExp}

The molecule of interest for our experiments will be Hsp90, a classical example used by the community to simulate biomolecular changes in cryo-EM (\cite{seitz2019simulation}).
In a nutshell, the heat shock protein 90 is a molecule assisting other proteins to fold properly. Images are generated with the \texttt{cryojax} library (\cite{o2026cryojax}). For each candidate
structure we read centered, hydrogen-free atomic coordinates from a PDB file and
build the electron-scattering potential as a sum of per-atom Gaussians from the
tabulated Peng scattering factors, retaining atomic $B$-factors. A
Gaussian-mixture integrator (with error-function pixel integration) then yields
the clean projection on a grid of $N = 128\times128$ pixels at $2.0$~\AA{}/pixel
and $300$~kV, with no PSF or translation and a single fixed pose (Euler angles
$(45^\circ,90^\circ,90^\circ)$) shared across all structures. Two images
therefore differ only through their conformation and noise, which is the
identity-forward-operator setting analyzed above. Let $\mathcal{M}$ be a cosine-tapered circular mask of radius $42$ pixels (one
third of the image width, one-pixel roll-off). Each clean projection $I_m$ is
standardized within the mask to unit signal variance,
\[
  \tilde I_m = \frac{I_m - \mu_{\mathcal{M}}(I_m)}{s_{\mathcal{M}}(I_m)},
\]
and the observed image for a sampled conformation $m$ is
\(
  y = \sqrt{\mathrm{SNR}}\,\tilde I_m + \varepsilon,
  \) \( \varepsilon \sim \mathcal{N}(0,I_N).
\)
The noise is thus i.i.d.\ standard Gaussian ($\sigma^2 = 1$ per pixel) and the
signal scales as $\sqrt{\mathrm{SNR}}$, so $\mathrm{SNR}$ is the in-mask
signal-to-noise power ratio and small $\mathrm{SNR}$ is the per-pixel high-noise regime of
interest. Example images are shown in Figure \ref{fig: EM}b.

The ground truth distribution is generated by picking conformations at random among the frames of the trajectory and assigning each to the nearest candidate structure; in some cases, we chose a uniform distribution while in other cases a bimodal one, see the orange line on the figures below.

\subsection{Reference structures not covering the space}
\label{sec: ConfSpace}

Recall that Figure~\ref{fig: EM} showed how early stopping in the EM acted as a regularizer while the exact solution to the reweighting optimization problem is spiky. We now turn to addressing the question of the importance of the candidates points covering the space well enough, which is the setting of Theorem~\ref{thm: LowBd}.
We first consider the set of reference structures with the motion shown in Figure~\ref{fig: EM}a to generate the data with varying signal-to noise ratio and number of sample points. To reflect the possibility that the set of reference structures might not be complete enough,  we remove the structures with opening angle between 15 and 20 degrees in the set of candidates.

On Figure~\ref{fig: Landscape}, one clearly sees that the candidate structure closer to the structures that have no representative in the dataset gets too much weight, which matches intuition. Also, even though less pronounced, a larger sample size yields a distribution that is closer to the uniform, while decreasing the signal-to-noise ratio also deteriorates the recovery (Figure~\ref{fig: Landscape}b).

\begin{figure}[ht!]
		\begin{center}
\includegraphics[width=0.99\textwidth]{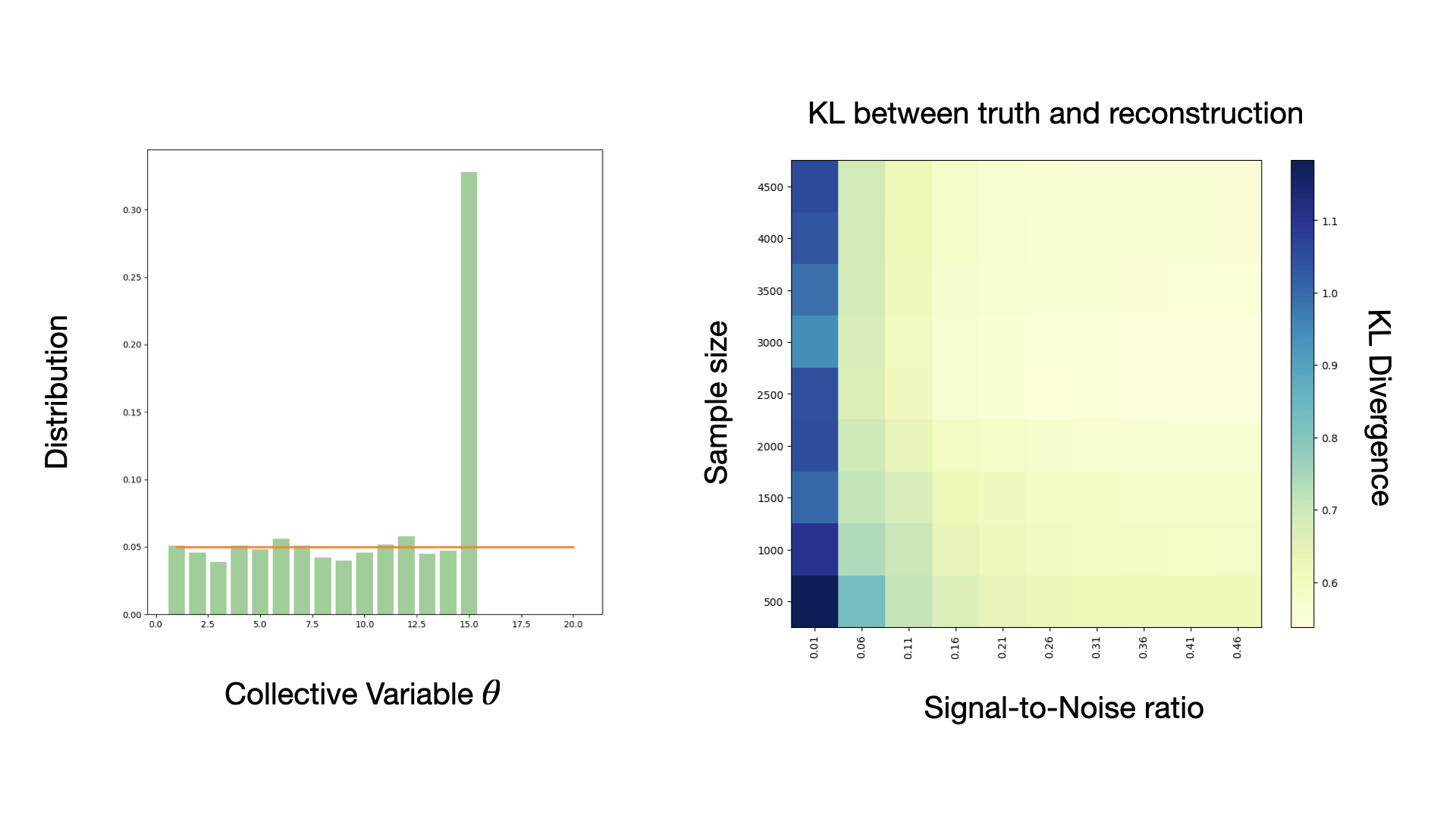}
	\caption{Output of the optimization procedure when the reference structures do not cover a part of the conformational space, while the ground truth is uniform over a larger set of structures.
    {\bf Left:} Comparison of the exact solutions of the problem for a sample size $n=1000$ and SNR=2. {\bf Right:}  $\operatorname{KL}(\widehat\alpha\|\alpha_0)$, averaged over 10 replications of the same random setting. For larger sample sizes and signal-to-noise ratio, the KL diminishes. It does not converge to zero as the set of candidate structures is not covering part of the conformational space. For very small SNR, the spikiness translates into higher KL-divergence.}
	\label{fig: Landscape}	
    \end{center}
\end{figure}

\subsection{Recovery of low-population states}
\label{sec: Transient}

In practice, even if the reference structures cover the entire conformational space, there will be regions in which the density of the structures is larger and regions of lower density (i.e., low-populated states). This is what we explore in the following example. As the spacing between the bins in the middle suggest, there are less structures in the middle of the conformational trajectory than at the boundaries. The actual motion is an opening, similar to that of Figure~\ref{fig: EM}a.
The optimal weights given a dataset are optimized using an exact convex solver and the results are reported in Figure~\ref{fig: lowDens}. In this example, as the reference structures do not cover the space well enough compared with the ground truth, they receive a much larger weight.

\begin{figure}[ht!]
\begin{center}
\includegraphics[width=0.95\textwidth]{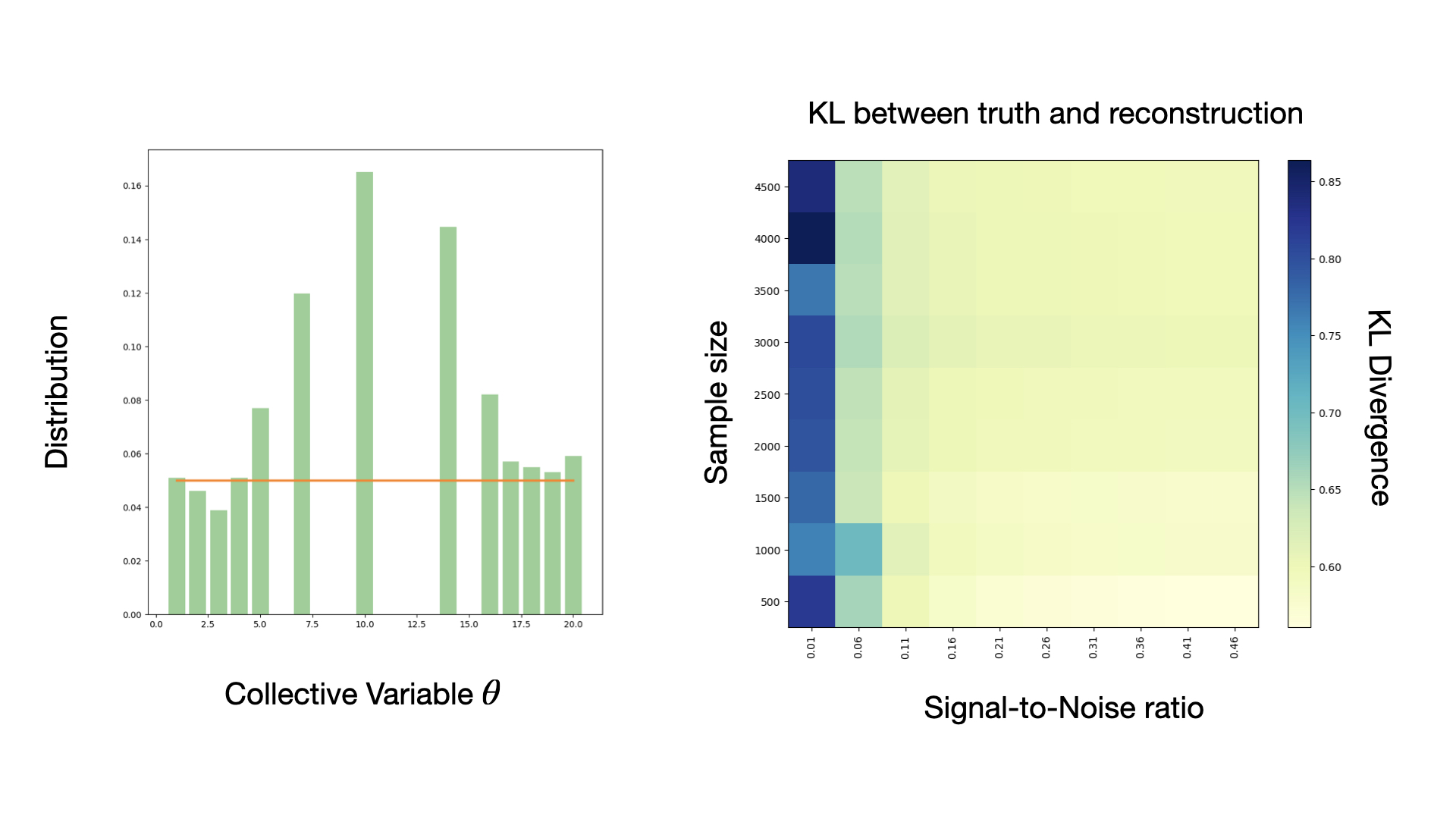}
\caption{\label{fig: lowDens} Output of the optimization procedure when the reference structures coarsely cover the central part of the conformational space, while the ground truth is uniform over a larger set of structures.
    {\bf Left:} Comparison of the exact solutions of the problem for a sample size $n=1000$ and SNR=2. {\bf Right:}  $\operatorname{KL}(\widehat\alpha\|\alpha_0)$, averaged over 10 replications of the same random setting. For larger sample sizes and signal-to-noise ratio, the KL diminishes. It does not converge to zero as the set of candidate structures is not covering part of the conformational space. For very small SNR, the spikiness translates into higher KL-divergence.}
   \end{center}
\end{figure}

\subsection{Where do the spikes appear?}
\label{sec: NumExpSpikes}
Section~\ref{sec: Replica} showed how empirical frequencies and candidate
geometry enter the objective. Figures~\ref{fig: Spikes} and
\ref{fig: Spikes2} illustrate that frequencies alone do not predict the
reconstructed weights. Indeed, the fact that a structure has been sampled (red) more than the ground truth does not imply that the reconstructed weight (green) is above the ground truth as well.
The observed spike locations reflect the joint effect of candidate
geometry, empirical frequencies and noise. 

\begin{figure}[ht!]
	\begin{center}
		\begin{tabular}{@{}c@{}c}\includegraphics[width=0.4\textwidth, height=45mm]{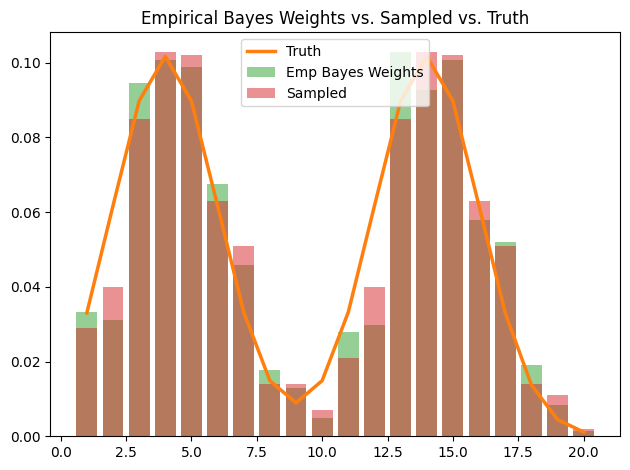}&
		\includegraphics[width=0.4\textwidth, height=45mm]{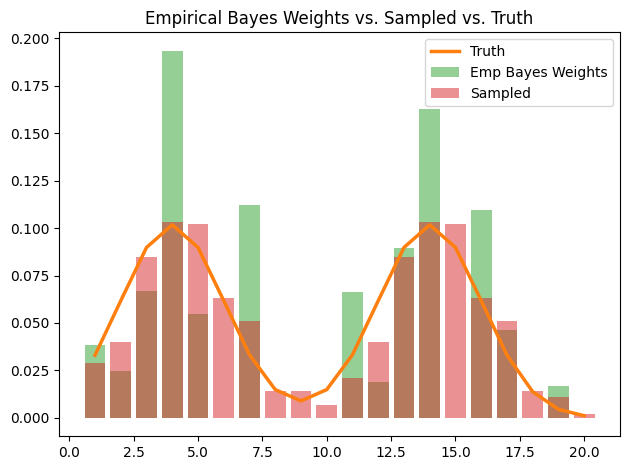}		\\	
		\end{tabular}
	\caption{Comparison between the reconstructed weights, the actual frequencies in the data and the ground truth. From left to right, the sample size is $n=1000$ and the SNR is 0.1 and 0.01, respectively.  }
	\label{fig: Spikes}	
    \end{center}
\end{figure}

\begin{figure}[ht!]
	\begin{center}
		\begin{tabular}{@{}c@{}c}\includegraphics[width=0.4\textwidth, height=45mm]{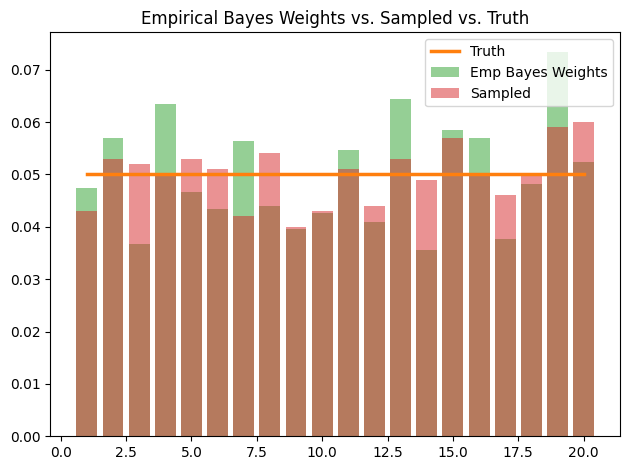}&
		\includegraphics[width=0.4\textwidth, height=45mm]{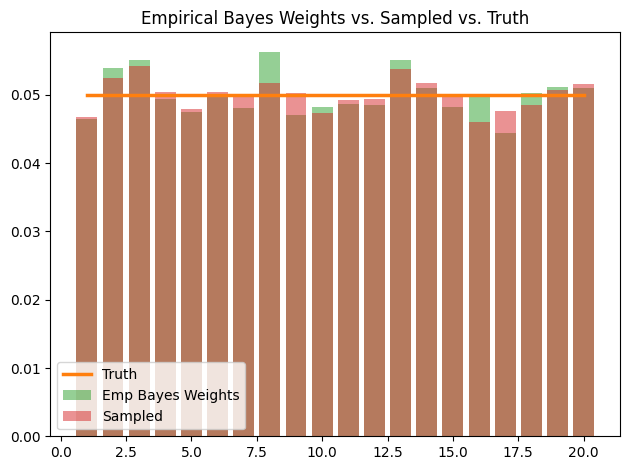}		\\	
		\end{tabular}
	\caption{Comparison between the reconstructed weights, the actual frequencies in the data and the ground truth. From left to right, the sample sizes are $n=1000$ and $n=10 000$, and the SNR is 0.1.}
	\label{fig: Spikes2}	
    \end{center}
\end{figure}

\section{Conclusions}
\label{sec: conclu}

In this work, we analyzed cryo-EM ensemble reweighting over a fixed set of candidate structures as a computational surrogate for a nonparametric maximum-likelihood estimation problem. We developed theoretical results and complemented them with numerical experiments illustrating the behavior of the recovered weights.

Our first conclusion is that the discretization induced by fixing a set of candidate structures is more than a mere approximation and is not a harmless computational convenience: the interplay between the noise level and the coverage of the candidate set of structures fundamentally controls identifiability, conditioning, and interpretability, as also discussed in \cite{mattingly2026measurement}. Candidate structures that are closer than the noise scale induce ill-conditioning, thereby making the weights sensitive to small perturbations. Conversely, when the candidates do not cover the space well enough, there is an unavoidable bias. These insights explain why spiky solutions are neither an anomaly nor necessarily conveying important information about the data; rather, they are an expected outcome of maximum-likelihood fitting on a discrete support.

We also showed that the choice of algorithm (and its parameters) has an important impact on the recovered solutions. In particular, EM slowdowns arise precisely in the near-nonidentifiable regime where structures are similar. Its tendency toward increasingly spiky fits motivates viewing the iteration count as an implicit regularization parameter. Importantly, this means that early stopping might be beneficial. The intricate nonlinear nature of the recursion prevents providing a clear theoretical rule, but we proposed principled stopping diagnostic tools. Future work may benefit from drawing analogs to regularization and early stopping in Richardson-Lucy iteration for image deconvolution (see for example~\cite{bertero2009image}), a close historical analog to our problem.

Overall, our work suggests shifting emphasis from solving the reweighting problem to a more nuanced approach in which one iteratively crafts the candidate set and carefully investigates the early-stopping potential for each dataset analyzed,  extending simulation frameworks such as in \cite{silva2026cryo}. Extending to the full cryo-EM model, i.e., including the pose and projection parts of the forward model, as well as reconstruction, is a step we aim to take in the future.

\section*{Acknowledgment}
Research reported in this publication was supported by the National Institute of General Medical Sciences of the National Institutes of Health under Award Number R35GM157226. The content is solely the responsibility of the authors and does not necessarily represent the official views of the National Institutes of Health. 
The work was also supported by the Alfred P. Sloan Foundation (FG-2023-20853), and the Simons Foundation (1288155). The Flatiron Institute is a division of the Simons Foundation.

\newpage
 \printbibliography
\appendix
\markboth{}{}

\newpage
\section{Proofs}
\label{sec: Proofs}

\begin{proof}[Proof of Proposition~\ref{prop: GradHessKKT}]
Direct differentiation gives the two formulas, and
\[
-v^\top\nabla^2F(\alpha)v
=\frac1n\sum_{i=1}^n
\frac{(K_{i,\bullet}^\top v)^2}{s_i(\alpha)^2}\geq0.
\]
This proves the curvature claims. If two maximizers had different fitted
vectors, strict concavity of $z\mapsto n^{-1}\sum_i\log z_i$ would give a
contradiction at their midpoint.
\end{proof}

\begin{proof}[Proof of Proposition~\ref{prop: IllCond}]
Write $K^{(h)}$ for the likelihood matrix defining $F_h$ and let 
$s^{(h)}(\alpha):=K^{(h)}\alpha$.
For $v=(e_k-e_\ell)/\sqrt2$, the Rayleigh quotient and~\eqref{eq: Hess}
give
\begin{equation}
\label{eq: UBLambda}
\lambda_{\min}^{\mathcal T}\!\bigl(-\nabla^2F_h(\alpha)\bigr)
\leq\frac1{2n}\sum_{i=1}^n
\left(\frac{K^{(h)}_{ik}-K^{(h)}_{i\ell}}
{s_i^{(h)}(\alpha)}\right)^2.
\end{equation}
On a fixed neighborhood of zero, smoothness and positivity of the Gaussian
kernel give constants $C,c>0$ such that, uniformly in $h$, $i$, and $\alpha$,
\[
|K^{(h)}_{ik}-K^{(h)}_{i\ell}|\leq C\|h\|,
\qquad s_i^{(h)}(\alpha)\geq c.
\]
The bound \eqref{eq: UBLambda} is therefore $O(\|h\|^2)$ uniformly in $\alpha$.
Under the stated high-noise scaling, divide the numerator and denominator by
$K^{(h)}_{i\ell}$. The likelihood ratios of the fixed candidates are
$1+O_p(\sigma^{-1})$, while
$K^{(h)}_{ik}/K^{(h)}_{i\ell}-1=O_p(\|h\|/\sigma)$ locally uniformly in $h$.
This gives the final claim.
\end{proof}

\begin{proof}[Proof of Proposition~\ref{prop: Self-Conc}]
For $u\in\mathbb R^M$, put
\(
z_i:= K_{i,\bullet}^\top u / K_{i,\bullet}^\top\alpha.
\)
Direct differentiation gives
\[
D^2f(\alpha)[u,u]=\frac1n\sum_{i=1}^n z_i^2,
\qquad
D^3f(\alpha)[u,u,u]=-
\frac2n\sum_{i=1}^n z_i^3.
\]
Therefore,
\[
|D^3f(\alpha)[u,u,u]|
\leq
\frac2n\left(\sum_{i=1}^n z_i^2\right)^{3/2}
=2\sqrt n\bigl(D^2f(\alpha)[u,u]\bigr)^{3/2}.
\]
Multiplication by $n$ changes the self-concordance constant from
$2\sqrt n$ to $2$, and the displayed Hessian vanishes in direction $u$
precisely when $Ku=0$.
\end{proof}

\begin{proof}[Proof of Theorem~\ref{thm: LowBd}]

For \(1\leq k\leq \mathfrak{m}+1\), define the symmetric
moment-difference tensor
\[
  T_k(\alpha)
  :=
  \int x^{\otimes k}\,
  d\bigl(\mathrm{P}_M^\alpha-\rho\bigr)(x).
\]
Matching all coordinate moments of total degree \(k\) is equivalent to
\(T_k(\alpha)=0\). By maximality of \(\mathfrak{m}\), no
\(\alpha\in\Delta_{M-1}\) makes
\[
  T_1(\alpha),\ldots,T_{\mathfrak{m}+1}(\alpha)
\]
all zero. These tensors depend continuously on \(\alpha\), and the
simplex is compact. Consequently,
\begin{equation}
  \delta
  :=
  \min_{\alpha\in\Delta_{M-1}}
  \left\{
    \sum_{k=1}^{\mathfrak{m}+1}
    \|T_k(\alpha)\|_F^2
  \right\}^{1/2}
  >0.
  \label{eq:uniform-moment-gap}
\end{equation}

Let
\[
  \Delta_\alpha(u)
  :=
  \mathcal{F}\bigl[\mathrm{P}_M^\alpha\bigr](u)
  -
  \mathcal{F}[\rho](u),
\]
where, for a finite measure \(\mu\),
\[
  \mathcal{F}[\mu](u)
  :=
  \int e^{i\langle u,x\rangle}\,d\mu(x).
\]
The supports of \(\rho\) and of every \(\mathrm{P}_M^\alpha\) lie in
one fixed ball. Taylor's formula for \(e^{i\langle u,x\rangle}\)
therefore gives, uniformly in \(\alpha\) and for \(\|u\|\leq 1\),
\begin{equation}
  \Delta_\alpha(u)
  =
  \sum_{k=1}^{\mathfrak{m}+1}
  \frac{i^k}{k!}
  \left\langle T_k(\alpha),u^{\otimes k}\right\rangle
  +
  E_\alpha(u),
  \qquad
  |E_\alpha(u)|
  \leq
  C_1\|u\|^{\mathfrak{m}+2}.
  \label{eq:uniform-characteristic-expansion}
\end{equation}

Let $B_1:=\{z\in\mathbb R^d:\|z\|\leq1\}$. The map
\[
  (A_1,\ldots,A_{\mathfrak{m}+1})
  \longmapsto
  \sum_{k=1}^{\mathfrak{m}+1}
  \frac{i^k}{k!}
  \left\langle A_k,z^{\otimes k}\right\rangle
\]
is injective on
\(
  \bigoplus_{k=1}^{\mathfrak{m}+1}
  \operatorname{Sym}^k(\mathbb{R}^d).
\)
Hence the \(L^2(B_1)\)-norm of the resulting polynomial is equivalent
to the Euclidean norm of its coefficients. It follows from
\eqref{eq:uniform-moment-gap} that, for every \(\sigma\geq 1\),
\begin{align}
  &\left\|
    \sum_{k=1}^{\mathfrak{m}+1}
    \frac{i^k\sigma^{-k}}{k!}
    \left\langle T_k(\alpha),z^{\otimes k}\right\rangle
  \right\|_{L^2(B_1)}
\geq
  c_1
  \left\{
    \sum_{k=1}^{\mathfrak{m}+1}
    \sigma^{-2k}\|T_k(\alpha)\|_F^2
  \right\}^{1/2}
  \geq
  c_1\delta\,
  \sigma^{-(\mathfrak{m}+1)}.
  \label{eq:polynomial-norm-lower-bound}
\end{align}

Set \(u=z/\sigma\). By
\eqref{eq:uniform-characteristic-expansion}, the \(L^2(B_1)\)-norm of
the remainder satisfies
\[
  \bigl\|E_\alpha(z/\sigma)\bigr\|_{L^2(B_1)}
  \leq
  C_2\sigma^{-(\mathfrak{m}+2)}
\]
uniformly in \(\alpha\). Therefore,
\begin{align*}
  \bigl\|\Delta_\alpha(z/\sigma)\bigr\|_{L^2(B_1)}
  &\geq
  c_1\delta\,\sigma^{-(\mathfrak{m}+1)}
  -
  C_2\sigma^{-(\mathfrak{m}+2)}.
\end{align*}
Choosing
\(
  \sigma_0
  \geq
  \max\left\{
    1,\frac{2C_2}{c_1\delta}
  \right\}\), it holds , for every \(\sigma\geq\sigma_0\), uniformly in
\(\alpha\in\Delta_{M-1}\), that
\[
  \bigl\|\Delta_\alpha(z/\sigma)\bigr\|_{L^2(B_1)}
  \geq
  \frac{c_1\delta}{2}
  \sigma^{-(\mathfrak{m}+1)}.
\]
After squaring, this gives
\begin{equation}
  \int_{B_1}
  \left|\Delta_\alpha(z/\sigma)\right|^2\,dz
  \geq
  c_2\sigma^{-2(\mathfrak{m}+1)}.
  \label{eq:characteristic-lower-bound}
\end{equation}

Since
\[
  \mathcal{F}[p_\sigma](u)
  =
  e^{-\sigma^2\|u\|^2/2},
\]
Plancherel's identity and the change of variables \(z=\sigma u\)
yield
\begin{align}
  \left\|
    p_\sigma*\mathrm{P}_M^\alpha-p_\sigma*\rho
  \right\|_2^2
  &=
  c_d
  \int_{\mathbb{R}^d}
  e^{-\sigma^2\|u\|^2}
  |\Delta_\alpha(u)|^2\,du
  \notag\\
  &=
  c_d\sigma^{-d}
  \int_{\mathbb{R}^d}
  e^{-\|z\|^2}
  |\Delta_\alpha(z/\sigma)|^2\,dz
  \notag\\
  &\geq
  c_de^{-1}\sigma^{-d}
  \int_{B_1}
  |\Delta_\alpha(z/\sigma)|^2\,dz
  \notag\\
  &\geq
  c_3
  \sigma^{-d-2(\mathfrak{m}+1)}.
  \label{eq:smoothed-l2-lower-bound}
\end{align}

Both $p_\sigma*\rho$ and $p_\sigma*\mathrm P_M^\alpha$ are bounded above by
$B_\sigma:=(2\pi\sigma^2)^{-d/2}$. The KL--Hellinger inequality gives
\[
\operatorname{KL}(p_\sigma*\rho\|p_\sigma*\mathrm P_M^\alpha)
\geq\frac1{4B_\sigma}
\|(p_\sigma*\rho)-(p_\sigma*\mathrm P_M^\alpha)\|_2^2.
\]
Combining this with~\eqref{eq:smoothed-l2-lower-bound} gives, uniformly in
$\alpha$,
\(
\operatorname{KL}(p_\sigma*\rho\|p_\sigma*\mathrm P_M^\alpha)
\geq c\,\sigma^{-2(\mathfrak m+1)}.
\)
\end{proof}

\begin{proof}[Proof of Proposition~\ref{prop: Replica}]
Write
\[
g_\alpha(y):=\log\{(p_\sigma*\mathrm P_M^\alpha)(y)\}.
\]
If $Y=x_k+\sigma Z$, write $\mathbb E_k$ for expectation under its law.
Finiteness of the candidate grid gives, uniformly in $\alpha,\beta$ and $k$,
\[
|g_\alpha(Y)|\leq C(1+\|Z\|^2),
\qquad
|g_\alpha(Y)-g_\beta(Y)|
\leq Ce^{C\|Z\|}\|\alpha-\beta\|_1.
\]
Indeed, the mixture density is a weighted average of the component densities,
so it lies between their minimum and maximum; this gives the first bound. For
the second, the mean-value inequality for the logarithm gives
\[
|g_\alpha(Y)-g_\beta(Y)|
\leq
\frac{\max_m p_\sigma(Y\mid x_m)}{\min_m p_\sigma(Y\mid x_m)}
\|\alpha-\beta\|_1.
\]
For Gaussian components with the same variance, the ratio on the right is
bounded by $Ce^{C\|Z\|}$. Both bounds are square integrable. The standard
maximal bound for this fixed-dimensional Lipschitz class therefore gives,
conditionally on the labels,
\[
\sup_{\alpha\in\Delta_{M-1}}
\left|
\frac1{N_k}\sum_{i:z_i=k}g_\alpha(Y_i)
-\mathbb E_k g_\alpha(Y)
\right|=O_p(N_k^{-1/2}).
\]
Moreover,
\[
\mathbb E_k g_\alpha(Y)
=-\frac d2\{1+\log(2\pi\sigma^2)\}
-\operatorname{KL}\!\left(p_\sigma(\cdot\mid x_k)\|
p_\sigma*\mathrm P_M^\alpha\right).
\]
Summing over the fixed number of clusters and using
$N_k/n\to\pi_k>0$ proves~\eqref{eq:clusterwise-kl-representation}.
\end{proof}

\begin{proof}[Proof of Proposition~\ref{prop: Prox}]
Up to an $\alpha$-independent constant,
\[
Q(\alpha\mid\alpha^{(k)})
=\sum_j\alpha_j^{(k+1)}\log\alpha_j,
\]
whose unique simplex maximizer is $\alpha^{(k+1)}$. Expanding the KL
divergence in $\mathcal D$ gives~\eqref{eq: Prox}.
\end{proof}

\begin{proof}[Proof of Proposition~\ref{prop: EqReg}]
Put $L:=T_\sigma+1$, $G_\sigma(\alpha):=\Phi_n(\alpha)-\alpha$, and
\[
\delta_\sigma:=\max_{i,j,\ell}
\left|\log\frac{K_{ij}}{K_{i\ell}}\right|,
\qquad
\varepsilon_\sigma:=e^{\delta_\sigma}-1.
\]
The Gaussian likelihood gives
\[
\log\frac{K_{ij}}{K_{i\ell}}
=\frac{\langle Z_i,x_j-x_\ell\rangle}{\sigma}
-\frac{\|x_j-\mu_i\|^2-\|x_\ell-\mu_i\|^2}{2\sigma^2}.
\]
Let
\[
W:=\max_{\substack{i,\,j<\ell\\x_j\neq x_\ell}}
\frac{|\langle Z_i,x_j-x_\ell\rangle|}{\|x_j-x_\ell\|}.
\]
The preceding identity and polarization give
\[
\delta_\sigma
\leq \frac{R_x}{\sigma}W+\frac{A_xR_x}{\sigma^2}
\leq a_\sigma\max(W,1).
\]
Moreover,
\[
\mathbb P(W>t)\leq nM(M-1)e^{-t^2/2}.
\]
Since the law of $W$ does not depend on $\sigma$, this bound is uniform for
the continuum limit $\sigma\to\infty$, and therefore
$\delta_\sigma=O_p(a_\sigma)$. In particular,
$\mathbb P(\delta_\sigma\leq1)\to1$, and on this event
$\varepsilon_\sigma\leq(e-1)\delta_\sigma$. Hence
$\varepsilon_\sigma=O_p(a_\sigma)$.
Since $s_i(\alpha)$ is a convex combination of the entries in row $i$,
\[
e^{-\delta_\sigma}\leq\frac{K_{ij}}{s_i(\alpha)}
\leq e^{\delta_\sigma},
\qquad
\left|\frac{K_{ij}}{s_i(\alpha)}-1\right|
\leq\varepsilon_\sigma.
\]
It follows that
\[
[G_\sigma(\alpha)]_j
=\frac{\alpha_j}{n}\sum_i
\left(\frac{K_{ij}}{s_i(\alpha)}-1\right),
\qquad
\sup_{\alpha\in\Delta_{M-1}}\|G_\sigma(\alpha)\|_2
\leq\varepsilon_\sigma.
\]
For a tangent vector $v\in\mathcal T$, direct differentiation and
$\sum_m v_m=0$ give
\[
[\nabla G_\sigma(\alpha)v]_j
=\frac1n\sum_i\left[
\left(\frac{K_{ij}}{s_i(\alpha)}-1\right)v_j
-r_j(y_i;\alpha)\sum_m
\left(\frac{K_{im}}{s_i(\alpha)}-1\right)v_m
\right].
\]
Since each $r(y_i;\alpha)$ is a probability vector, the preceding row-wise
bound implies
\[
\|\nabla G_\sigma(\alpha)v\|_2
\leq(1+\sqrt M)\varepsilon_\sigma\|v\|_2.
\]
Integrating along the segment joining $\alpha$ and $\beta$ therefore gives
\[
\|G_\sigma(\alpha)-G_\sigma(\beta)\|_2
\leq C_M\varepsilon_\sigma\|\alpha-\beta\|_2,
\qquad C_M:=1+\sqrt M.
\]
The bound
$\sup_{\alpha\in\Delta_{M-1}}\|G_\sigma(\alpha)\|_2
\leq\varepsilon_\sigma$ gives
$\|\alpha^{(t)}-u\|_2\leq t\varepsilon_\sigma$.
Moreover, the EM recursion yields
\[
\alpha^{(L)}-u-LG_\sigma(u)
=\sum_{t=0}^{L-1}
\bigl\{G_\sigma(\alpha^{(t)})-G_\sigma(u)\bigr\},
\]
and
\[
\|\alpha^{(L)}-u-LG_\sigma(u)\|_2
\leq C_M\varepsilon_\sigma
\sum_{t=0}^{L-1}t\varepsilon_\sigma
\leq \frac{C_M}{2}L^2\varepsilon_\sigma^2.
\]

Write $\widetilde\alpha:=\widetilde\alpha_{\sigma,T_\sigma}$. Since
$\operatorname{KL}(u\|\alpha)$ diverges at the boundary and is strictly
convex in the relative interior, the penalized objective has a unique interior
maximizer. Its stationarity equations are
\[
[\nabla F(\widetilde\alpha)]_j
+\frac{1}{LM\widetilde\alpha_j}=\nu.
\]
Multiplying by $\widetilde\alpha_j$ and summing gives
$\nu=1+L^{-1}$. Multiplying the $j$th stationarity equation by
$\widetilde\alpha_j$ and using
$[\Phi_n(\alpha)]_j=\alpha_j[\nabla F(\alpha)]_j$ gives
\[
\bigl[\Phi_n(\widetilde\alpha)\bigr]_j+\frac1{LM}
=(1+L^{-1})\widetilde\alpha_j,
\quad \text{ as well as } \quad 
L[G_\sigma(\widetilde\alpha)]_j
=\widetilde\alpha_j-\frac1M.
\]
Thus $\widetilde\alpha-u=LG_\sigma(\widetilde\alpha)$, and consequently
\[
\|\widetilde\alpha-u\|_2\leq L\varepsilon_\sigma,
\qquad
\|\widetilde\alpha-u-LG_\sigma(u)\|_2
\leq C_ML^2\varepsilon_\sigma^2.
\]
Together with the preceding bound on
$\|\alpha^{(L)}-u-LG_\sigma(u)\|_2$, this shows that the numerator in
\eqref{eq:empirical-em-regularization-rate} is $O_p(L^2a_\sigma^2)$.

It remains to control the denominator. With
$\bar x:=M^{-1}\sum_jx_j$ and $\bar Z:=n^{-1}\sum_iZ_i$, a first-order
softmax expansion at the uniform weights gives
\[
[G_\sigma(u)]_j
=\frac{\langle\bar Z,x_j-\bar x\rangle}{M\sigma}
+O_p(\sigma^{-2}).
\]
Since the grid is fixed, $a_\sigma\sim R_x/\sigma$. The bound on
$\alpha^{(L)}-u-LG_\sigma(u)$ and $La_\sigma\to0$ therefore give
\[
\frac{\alpha^{(L)}-u}{La_\sigma}
=\frac{G_\sigma(u)}{a_\sigma}+O_p(La_\sigma)
=\frac1{MR_x}
\bigl(\langle\bar Z,x_j-\bar x\rangle\bigr)_{j=1}^M+o_p(1).
\]
Because $R_x>0$, the map
$z\mapsto(\langle z,x_j-\bar x\rangle)_{j=1}^M$ has positive rank, so the
Gaussian vector on the right is nonzero almost surely. Consequently,
$La_\sigma/\|\alpha^{(L)}-u\|_2=O_p(1)$.
Combining this with the bound on the numerator gives
\[
\frac{\|\alpha^{(L)}-\widetilde\alpha\|_2}
{\|\alpha^{(L)}-u\|_2}
=\frac{\|\alpha^{(L)}-\widetilde\alpha\|_2}
{L^2a_\sigma^2}
\frac{La_\sigma}{\|\alpha^{(L)}-u\|_2}\,La_\sigma
=O_p(La_\sigma)=o_p(1). \qedhere
\]
\end{proof}

\begin{proof}[Proof of Proposition~\ref{prop: SelfConcGrad}]
Let $f=-F$, $\phi(s)=f(\beta+sd)$, and $q(s)=\sqrt{\phi''(s)}$.
If $q(0)=0$, then $Kd=0$ and the claim is immediate. Otherwise,
Proposition~\ref{prop: Self-Conc} gives
$|q'(s)|\leq\sqrt n\,q(s)^2$, whence
\[
\frac{q(0)}{1+s\tau}\leq q(s)\leq\frac{q(0)}{1-s\tau}.
\]
Squaring and integrating over $s\in[0,1]$ gives the two bounds in
\eqref{eq:self-concordant-gradient-bound}.
\end{proof}

\begin{proof}[Proof of Theorem~\ref{thm: NoiseMisspecification}]
Start by locally defining
\[
\delta_N:=\frac1{\sigma^2}
\max_{i,m,\ell}|\langle x_m-x_\ell,\eta_i\rangle|.
\]
For some $C_\eta$ independent of $\alpha$, the perturbed objective is
\[
\widetilde F(\alpha)=C_\eta+\frac1n\sum_i
\log\left(\sum_m\alpha_mK_{im}
\exp\!\left\{\frac{\langle x_m,\eta_i\rangle}{\sigma^2}\right\}\right).
\]
For each $i$, the exponents in the sum above differ by at most $\delta_N$.
Consequently, coordinatewise, the corresponding summands of the two
gradients have a ratio between $e^{-\delta_N}$ and $e^{\delta_N}$; the same
is true after averaging over $i$. The KKT conditions at $\widetilde\alpha$
then give
\[
\|\nabla F(\widetilde\alpha)
-\nabla\widetilde F(\widetilde\alpha)\|_2
\leq\sqrt M(e^{\delta_N}-1).
\]

Under~\eqref{eq:tangent-injectivity}, continuity and compactness give
$\mu_N>0$. With $d=\widetilde\alpha-\widehat\alpha$, strong concavity and
the optimality conditions give
\[
\mu_N\|d\|^2
\leq\langle\nabla F(\widehat\alpha)
-\nabla F(\widetilde\alpha),d\rangle
\leq\langle\nabla\widetilde F(\widetilde\alpha)
-\nabla F(\widetilde\alpha),d\rangle.
\]
Therefore
\[
\|d\|\leq\frac{\sqrt M}{\mu_N}(e^{\delta_N}-1).
\]
At the uniform vector $u$, restrict $-\nabla^2F(u)$ to the
$(M-1)$-dimensional space $\mathcal T$. The sum of its eigenvalues is
$M^2n^{-1}\sum_i\|r(y_i;u)-u\|^2$. Since $\mu_N$ is no larger than any of
these eigenvalues and $r(y_i;u)$ is a probability vector,
\[
(M-1)\mu_N
\leq M^2\frac1n\sum_i\|r(y_i;u)-u\|^2
\leq M(M-1).
\]
Here the last inequality uses
$\|r(y_i;u)-u\|^2=\|r(y_i;u)\|^2-1/M\leq1-1/M$.
Thus $\mu_N\leq M$. If $\delta_N\leq1$, the preceding bound and
$e^{\delta_N}-1\leq2\delta_N$ give
\[
\|d\|\leq\frac{2\sqrt M}{\mu_N}\delta_N
\leq\frac{2M}{\mu_N}\delta_N.
\]
If $\delta_N>1$, both maximizers lie in the simplex, whose diameter is
$\sqrt2$. Using also $\mu_N\leq M$ gives
\[
\|d\|\leq\sqrt2\leq\frac{2M}{\mu_N}\delta_N.
\]
Hence the last bound holds in both cases.
Finally,
\[
\mathbb E\langle x_m-x_\ell,\eta_i\rangle^2
\leq R_x^2\varepsilon_N^2.
\]
Chebyshev's inequality and a finite union bound give
$\delta_N=O_p(R_x\varepsilon_N/\sigma^2)$, which proves
\eqref{eq:misspecified-noise-argmax}.
\end{proof}

\begin{proof}[Proof of Theorem~\ref{thm: CLT}]
Write $m_\theta(y)=\log f_\theta(y)$ and
$\mathbb P_nh=n^{-1}\sum_{i=1}^nh(Y_i)$. Write
$Y=X+\sigma W$, where $X\sim\rho$ and $W\sim\mathcal N(0,I_d)$ are
independent. Bounded support of $\rho$ gives $\mathbb E\|Y\|^2<\infty$,
and
\[
\max_m|\log p_\sigma(Y\mid x_m)|\leq C(1+\|Y\|^2).
\]
Since $|m_\theta|$ is bounded by the left-hand side, compactness of
$\Theta$ and continuity give
\[
\sup_{\theta\in\Theta}
|\mathbb P_nm_\theta-\mathbb E m_\theta|\longrightarrow0
\quad\text{in probability}.
\]

Choose $k$ with $\alpha_k^*>0$. Since
$f_{\theta^*}\geq\alpha_k^*p_\sigma(\cdot\mid x_k)$,
\[
\|s_{\theta^*}(y)\|
\leq C\sum_m\frac{p_\sigma(y\mid x_m)}{p_\sigma(y\mid x_k)}.
\]
Each ratio on the right is the exponential of an affine function of $y$.
The random vector $Y=X+\sigma W$ has finite exponential moments in every
linear direction, so $I_{\theta^*}$ is finite. Moreover, for
$v\in\mathbb R^{M-1}$,
\[
v^\top I_{\theta^*}v
=\mathbb E\left[
\frac{\{\sum_{j<M}v_j(p_\sigma(Y\mid x_j)-p_\sigma(Y\mid x_M))\}^2}
{f_{\theta^*}(Y)^2}
\right].
\]
The density $p_\sigma*\rho$ is positive everywhere, and distinct Gaussian
translates are linearly independent, which can be seen  by taking
Fourier transforms. Hence this display can vanish only for $v=0$, proving
that $I_{\theta^*}$ is positive definite.

Linearity of $f_\theta$ gives the exact identity
\[
m_{\theta^*+h}(y)-m_{\theta^*}(y)
=\log\{1+h^\top s_{\theta^*}(y)\}
\]
for every feasible $h$.
If $\theta^*+v$ were another population maximizer, concavity would make the
population criterion constant on the segment joining the two maximizers,
whereas this identity and the finite second score moment give the strictly
negative second directional derivative $-v^\top I_{\theta^*}v$. Thus the population
maximizer is unique and the uniform law above gives
$\hat\theta_n\to\theta^*$ in probability.

Put $S_i=s_{\theta^*}(Y_i)$ and
\[
Z_n:=\sqrt n(\mathbb P_n-P_0)s_{\theta^*}.
\]
Finite second moments give
\[
\frac{\max_{i\leq n}\|S_i\|}{\sqrt n}=o_p(1),
\qquad
\frac1n\sum_{i=1}^nS_iS_i^\top\longrightarrow I_{\theta^*},
\]
and the multivariate central limit theorem gives
$Z_n\rightsquigarrow Z\sim
\mathcal N_{M-1}(0,I_{\theta^*}-bb^\top)$. Uniformly for $u$ in a fixed
compact set with $\theta^*+u/\sqrt n\in\Theta$, Taylor expansion in the
exact identity now gives
\begin{align*}
&n\big(\mathbb P_nm_{\theta^*+u/\sqrt n}
-\mathbb P_nm_{\theta^*}\big)\\
&\qquad=\sqrt n\,b^\top u+Z_n^\top u
-\frac12u^\top I_{\theta^*}u+o_p(1).
\end{align*}

Population optimality gives $b^\top u\leq0$ in every feasible local
direction. If $\lambda$ is the smallest eigenvalue of $I_{\theta^*}$, the
last display is therefore at most
$R\|Z_n\|-\lambda R^2/2+o_p(1)$ on the feasible sphere $\|u\|=R$.
Consistency, concavity, and then $R\to\infty$ give
$\sqrt n(\hat\theta_n-\theta^*)=O_p(1)$.

Because $\Theta$ is defined by linear inequalities, its root-$n$ local
feasible directions eventually satisfy exactly the first two conditions in
the definition of $\mathcal C$ on every compact set. Along a convergent
feasible sequence, the local criterion tends to $-\infty$ if $b^\top u<0$;
if $u\in\mathcal C$, its upper limit is at most
\[
Z^\top u-\frac12u^\top I_{\theta^*}u.
\]
Conversely, every $u\in\mathcal C$ is eventually feasible and attains this
limit. Tightness and the argmax argument for these concave criteria give
\eqref{eq:constrained-mle-limit}; positive definiteness of $I_{\theta^*}$
gives uniqueness of the limiting maximizer.

If $\theta^*$ is interior, first-order optimality gives $b=0$ and
$\mathcal C=\mathbb R^{M-1}$. The limit is then $I_{\theta^*}^{-1}Z$, with
covariance $I_{\theta^*}^{-1}$.

In the case where the model is correctly specified,
\[
b_j=\int\{p_\sigma(y\mid x_j)-p_\sigma(y\mid x_M)\}\,\diff y=0,
\]
so $\mathcal C$ is the tangent cone and
$Z\sim\mathcal N_{M-1}(0,I_{\theta^*})$. Moreover,
\[
\frac12u^\top I_{\theta^*}u-Z^\top u
=\frac12\|I_{\theta^*}^{1/2}u-I_{\theta^*}^{-1/2}Z\|^2
-\frac12Z^\top I_{\theta^*}^{-1}Z.
\]
Therefore, after multiplication by $I_{\theta^*}^{1/2}$, the minimizer is
the Euclidean projection of
$I_{\theta^*}^{-1/2}Z\sim\mathcal N_{M-1}(0,I_{M-1})$ onto
$I_{\theta^*}^{1/2}\mathcal C$.
\end{proof}

\section{Additional Stopping Criterion on IgG example}
\label{sec: igg}
 We now consider a different benchmark example to illustrate the
stopping-criterion issue discussed in Section~\ref{sec: EMSect}. The simulation is from the 100 structures of the IgG-1d example from the benchmark \cite{jeon2024cryobench}, but simulated with conformations taking a multimodal distribution as considered in \cite{evans2026counting} (Figure~\ref{fig: igg}).  
Unlike the Hsp90 example above, the EM algorithm overfits very quickly, within 1000 iterations. Here, the threshold criteria at $10^{-4}$ is too strict, unlike the previous example in Figure~\ref{fig: EM_with_thresholds}.
The iteration at which overfitting begins varies considerably between examples, and understanding this stopping time is an important direction for future work.
Figures~\ref{fig: EM_with_thresholds_igg} and
\ref{fig: igg_gaps_comparison} report the corresponding stopping diagnostics.
\begin{figure}
\begin{center}
\includegraphics[width=0.75\textwidth]{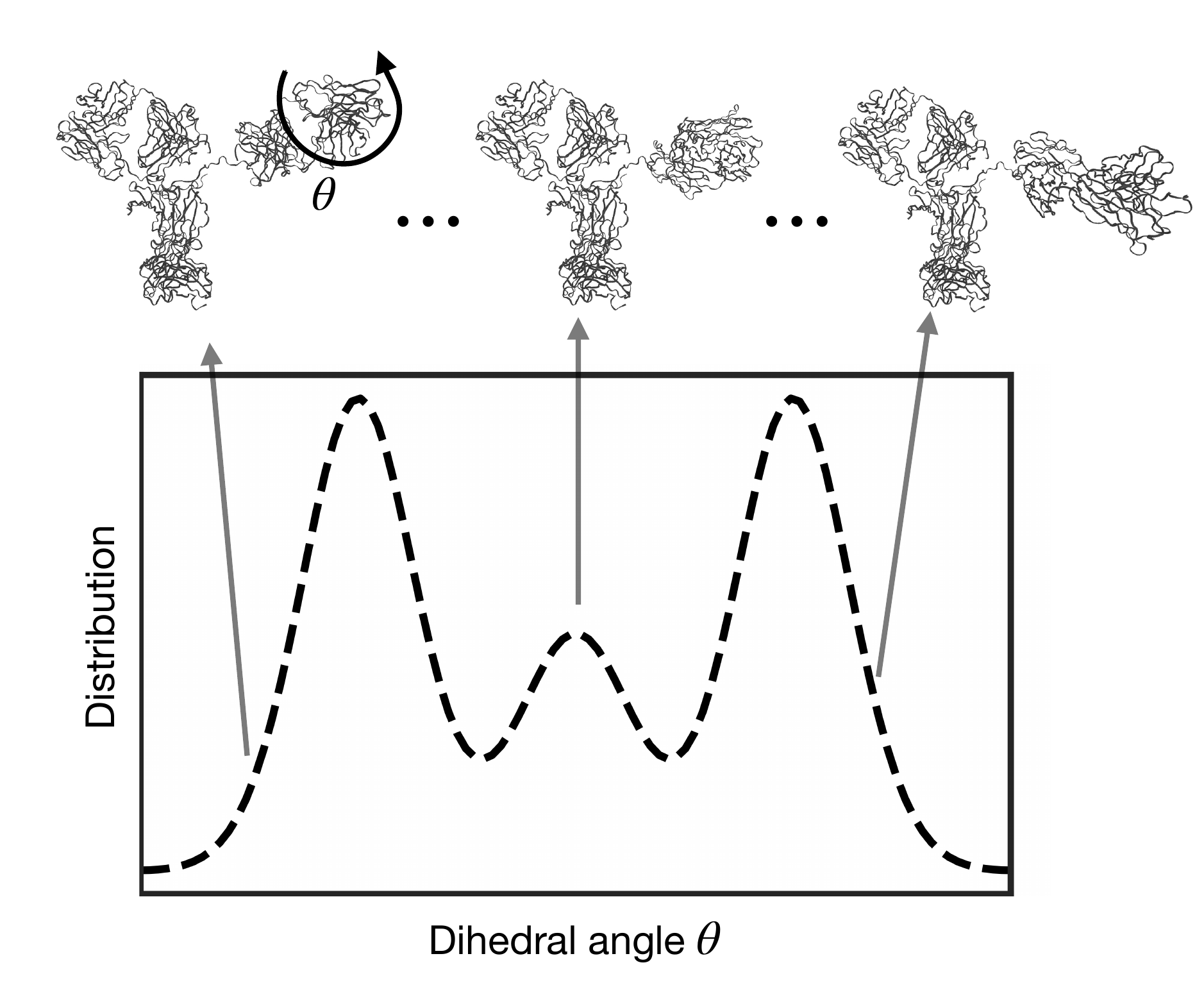}
\caption{\label{fig: igg}
Example structures of the IgG molecule used in \ref{sec: igg}, with diagram of the distribution of structures along the dihedral angle $\theta,$ as considered in \cite{evans2026counting}.
}
\end{center}
\end{figure}

\begin{figure}
\begin{center}
\includegraphics[width=0.99\textwidth]{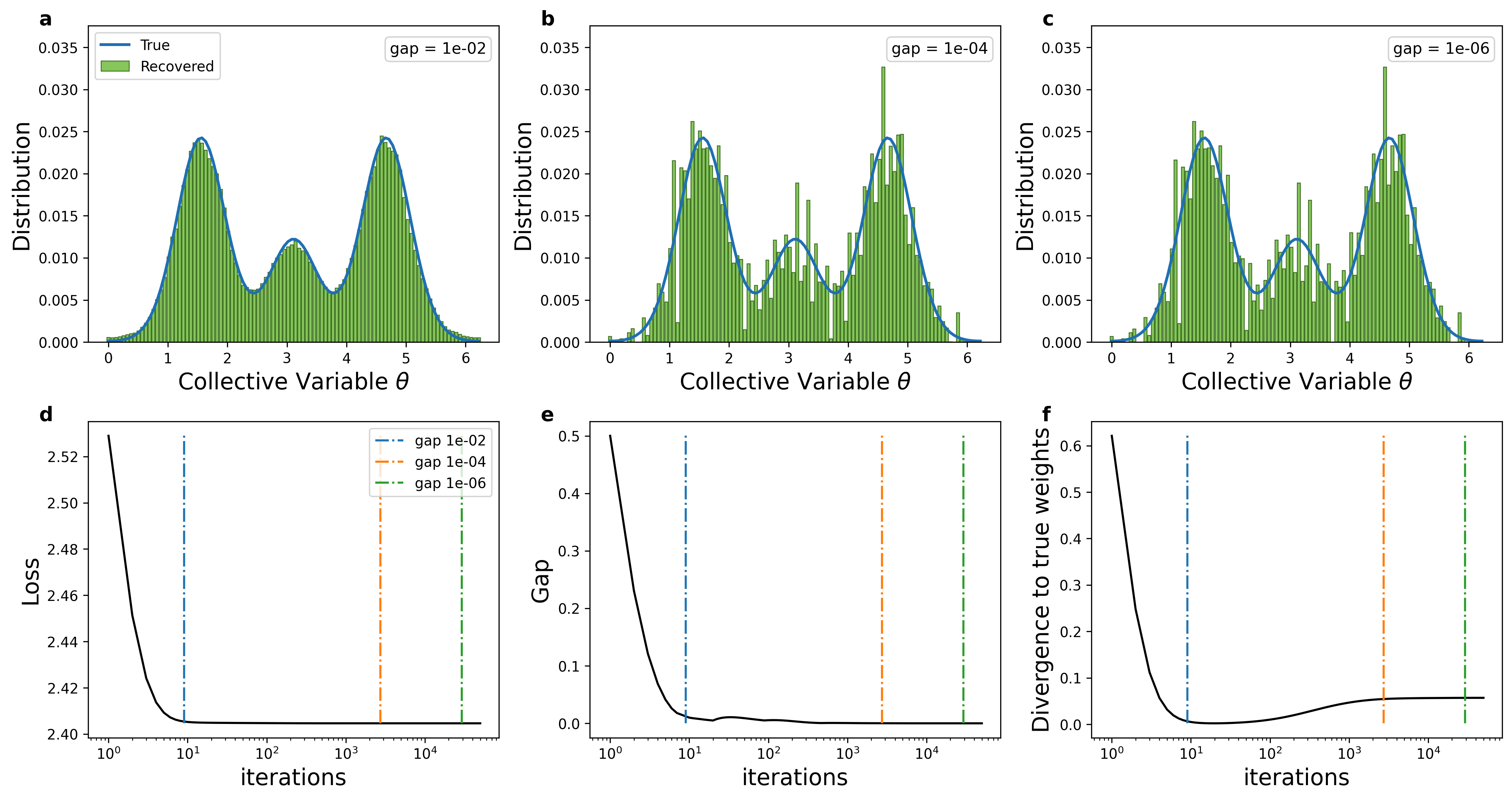}
\caption{\label{fig: EM_with_thresholds_igg}
Comparison of various thresholds for the gap $||\nabla F(\alpha_t) - 1 ||_{\infty}$ on the IgG example from \ref{sec: igg}.
{\bf a.-\bf c.} Estimated weights at  thresholds $10^{-2}, 10^{-4}, 10^{-6}$ respectively. {\bf d.-\bf f. } Loss, gap, and divergence to true weights  respectively, plotted against iteration number. The blue, orange, and green lines denote the iteration index where the $10^{-2}, 10^{-4}, 10^{-6}$ gap thresholds are reached.
}
\end{center}
\end{figure}

\begin{figure}
\begin{center}
\includegraphics[width=0.99\textwidth]{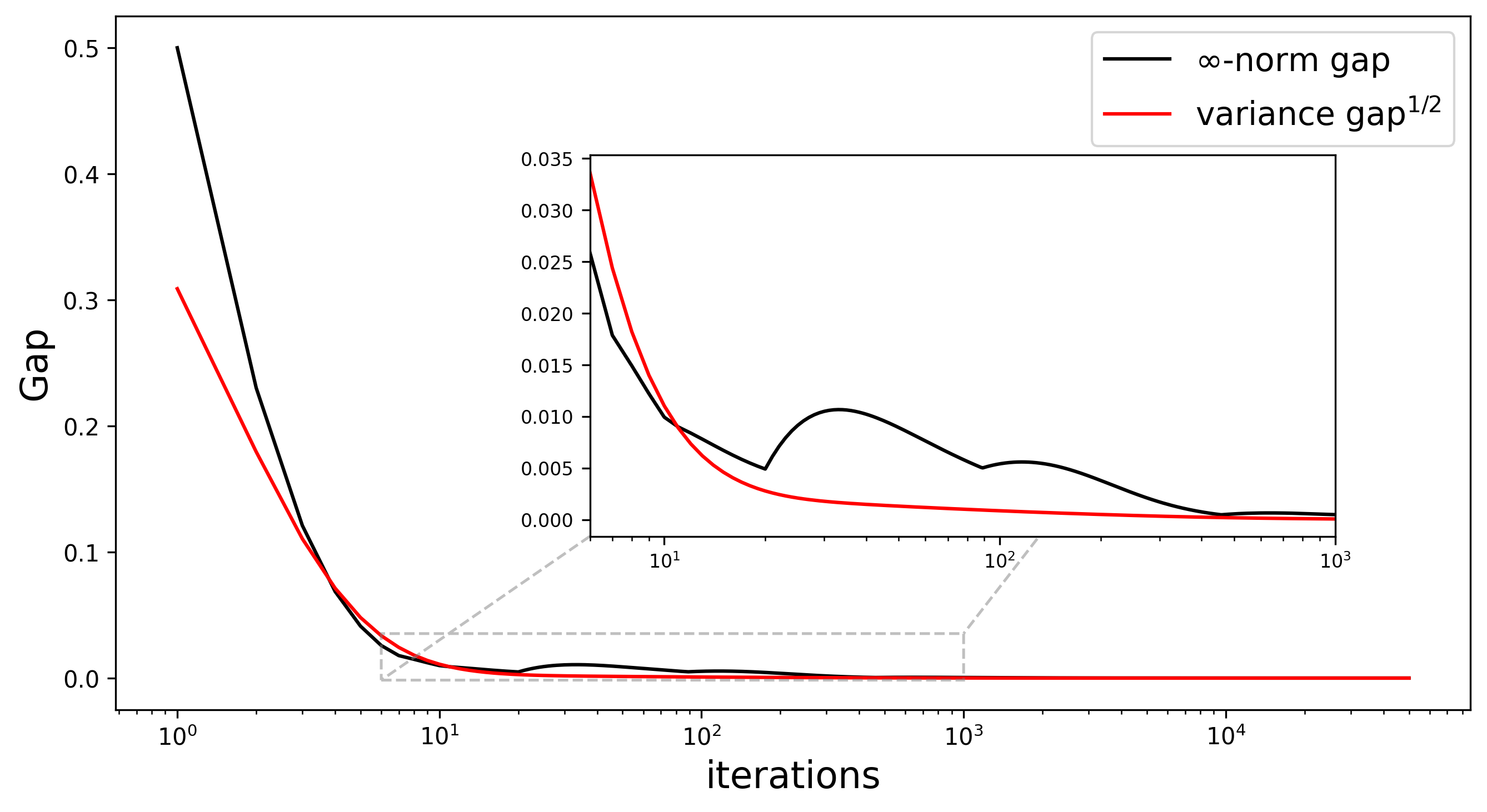}
\caption{\label{fig: igg_gaps_comparison}
 Comparison of the ``variance gap'' from Equation \eqref{eq: stopCrit} (red) against the $\infty-$norm gap $||\nabla F(\alpha^{(t)}) - 1_{M}||_{\infty},$ for the IgG example considered in \ref{sec: igg}, with the same data as figure~\ref{fig: EM_with_thresholds_igg}.
}
\end{center}
\end{figure}

\end{document}